\documentclass[Journal]{IEEEtran}

\ifCLASSINFOpdf
\else
   \usepackage{graphicx}
\fi
\usepackage{url}
\usepackage{amsmath,amssymb,exscale}
\usepackage[compress]{cite}

\usepackage{color}
\usepackage{graphicx}
\usepackage{algorithm}
\usepackage{algorithmic}
\usepackage[center]{caption}
\usepackage{verbatim}
\usepackage[bookmarks=false,hidelinks]{hyperref}
\usepackage{balance}

\usepackage{booktabs}
\usepackage{array}
\usepackage{multirow}
\usepackage{resizegather}
\usepackage{bbm}
\usepackage{subcaption}
\usepackage{lineno}

\usepackage{amsmath,amssymb,bm}
\usepackage{graphicx}
\usepackage{cite}
\usepackage{xcolor}

\newcommand{\revised}[1]{\textcolor{black}{#1}}

\usepackage{amsthm}
\newtheoremstyle{italic-head-normal-body}
  {\topsep}{\topsep}
  {\normalfont}
  {}
  {\itshape}
  {}
  {0.5em}
  {\thmname{#1}~\thmnumber{#2}: \thmnote{(#3)}}

\theoremstyle{italic-head-normal-body}
\newtheorem{theorem}{Theorem}
\newtheorem{proposition}[theorem]{Proposition}
\newtheorem{corollary}{Corollary}

\begin{document}

\captionsetup[figure]{name={Fig.},labelsep=period,singlelinecheck=off}
\graphicspath{{Images/}}
\renewcommand{\algorithmicrequire}{\bf{Input:}}
\renewcommand{\algorithmicensure}{\bf{Output:}}

\title{SemISAC: Semantic Integrated Sensing and Communications}

\author{Xiaoqi Zhang,~\IEEEmembership{Student Member,~IEEE,}
        J. Andrew Zhang,~\IEEEmembership{Senior Member,~IEEE},\\
        Zhongqin Wang, ~\IEEEmembership{Member,~IEEE},
        Chang Liu,~\IEEEmembership{Member,~IEEE,}\\
        Weijie Yuan,~\IEEEmembership{Senior Member,~IEEE,}
        Giuseppe Caire,~\IEEEmembership{Fellow,~IEEE,}
        and Geoffrey Ye Li,~\IEEEmembership{Fellow,~IEEE}

\thanks{X. Zhang, J. A. Zhang and Z. Wang are with the School of Electrical and Data Engineering, University of Technology Sydney, Sydney, NSW 2007, Australia (e-mail: xiaoqi.zhang@student.uts.edu.au; andrew.zhang@uts.edu.au; zhongqin.wang@uts.edu.au).}
\thanks{C. Liu is with the Department of Computer Science and Information Technology, La Trobe University, Melbourne, Victoria 3086, Australia (e-mail: C.Liu6@latrobe.edu.au).}
\thanks{W. Yuan is with the Department of Electrical and Computer Systems Engineering, Monash University, Clayton, VIC 3800, Australia (e-mail: weijie.yuan@monash.edu.au).}
\thanks{G. Caire is with the Chair of Communications and Information Theory, Technical University of Berlin, 10623 Berlin, Germany (e-mail: caire@tuberlin.de).}
\thanks{G. Y. Li is with the Department of Electrical and Electronic Engineering, Imperial College London, London SW7 2AZ, U.K. (e-mail: geoffrey.li@imperial.ac.uk).}

\thanks{(Corresponding author: J. Andrew Zhang.)}
}

\maketitle

\begin{abstract}

Conventional integrated sensing and communications (ISAC) systems primarily integrate communications and sensing through shared physical resources, without explicitly exploiting task-relevant semantic information. To move beyond such physical-level integration, we propose semantic ISAC (SemISAC), a general framework that unifies semantic communication (SemCom) and semantic sensing (SemS) to convey source meaning and acquire environmental meaning. Specifically, the transmitter combines source semantics and sensing task information with available side information to design the shared waveform and allocate radio resources, while the receiver-side communication and sensing task decoders recover the source meaning and infer the required environmental information, respectively. We also provide an information-theoretic interpretation to characterize the relationship between physical and task-relevant information and the resulting semantic trade-off in SemISAC. Building on this framework, we formulate the general SemISAC design problem and propose two realization methods, namely end-to-end (E2E) SemISAC optimization and modular SemISAC optimization. As a concrete realization, we apply modular SemISAC optimization to jointly design a learnable time--frequency (TF) precoder in an orthogonal frequency-division multiplexing (OFDM) system for representative SemCom and SemS tasks. Simulation results demonstrate that the proposed realization reduces sensing semantic distortion under a given communication requirement and achieves a more favorable communication--sensing trade-off than baseline designs.

\end{abstract}

\IEEEpeerreviewmaketitle

\section{Introduction}

Future wireless networks are expected to support intelligent services that depend on reliable connectivity, environmental awareness, distributed computation, and autonomous decision-making \cite{liu2022isac}. Applications such as autonomous transportation, robotics, unmanned platforms, digital twins, and immersive interaction therefore require a common radio platform capable of both delivering source information and acquiring environmental information \cite{wu2025humanEnvironmentalSensing,xiaoqi2024beamforming}. Integrated sensing and communication (ISAC) addresses this requirement by sharing spectrum, hardware, signal-processing modules, and transmitted waveforms between the two functions \cite{zhang2022jcas,xiaoqi2023enhanced}. However, physical-layer integration is insufficient for task-oriented services because conventional bit-level and estimation-level metrics do not directly characterize how transmitted source information and acquired environmental information support downstream inference, sensing, and decision-making \cite{wen2024taskiscc}.

Conventional communication systems pursue reliable source reconstruction under a given fidelity criterion. When only task-relevant content is needed from high-dimensional images, videos, or multimodal data, resources used to reconstruct other source components do not support the receiver task. Semantic communication (SemCom) addresses this mismatch by preserving source meaning or decision-relevant features rather than every transmitted bit. For example, DeepSC \cite{xie2021deepsc} jointly learns a semantic representation and channel adaptation for end-to-end semantic transmission, improving text-transmission robustness over noisy channels. An alternative information-bottleneck approach \cite{shao2022ib} jointly optimizes feature extraction, source coding, and channel coding to retain features that determine the downstream decision and deliver them under bandwidth and latency constraints.

Semantic sensing (SemS), by contrast, acquires task-relevant information from the physical environment. Conventional radar and wireless sensing designs estimate target existence, delay, Doppler shift, angle, range, or velocity and optimize detection probability, estimation error, or the Cram\'er--Rao bound (CRB) \cite{wu2025humanEnvironmentalSensing,xiaoqi2023learnCRB}. However, many downstream tasks do not require uniform fidelity across all environmental components. For example, the task-oriented integrated sensing, communication, and computation design in \cite{wen2024taskiscc} uses a discriminant-gain criterion to connect physical sensing conditions to inference performance and coordinate resources for multi-device edge inference. The SemS framework in \cite{xiaoqi2026semanticSensing,zhang2026semantic} instead introduces task-dependent relevance to prioritize the objects, propagation paths, and physical parameter dimensions that determine the required environmental task output. Both methods concentrate sensing resources on environmental information that contributes directly to the downstream task.

Existing studies that incorporate task objectives and semantic processing into sensing and ISAC can be organized into three research directions. \emph{One direction} introduces task objectives into conventional sensing, communication, and computation. For federated edge learning, \cite{liu2023ambient} uses a learning-convergence bound to coordinate sensing and communication resources with adaptive batch sizes, thereby accelerating model training under latency, energy, and sensing-quality constraints. In a task-oriented design for multi-device edge inference, \cite{zhuang2024iscc} combines broadband over-the-air computation with a max--min pairwise discriminant-gain criterion to jointly optimize sensing power, transmit precoding, and receive beamforming under limited energy.

Beyond this task-oriented coordination of conventional modules, \emph{another direction} combines sensing with semantic communication by processing sensing-derived information. Specifically, \cite{du2023semantic_sensing_comm} selects and compresses task-relevant sensing observations after data acquisition to reduce their storage and transmission overhead. In a multimodal setting, \cite{peng2025simac} employs semantic fusion to integrate radar and visual observations before semantic encoding and task-specific decoding. For extended-reality services, \cite{zhang2022xr} jointly considers sensing, rendering, and communication by exploiting the spatiotemporal distribution of environmental information to reduce the resulting traffic load.

\emph{A further direction} extends semantic processing across both sensing and communication. More recently, \cite{ren2025semantic_empowered_isac} introduces semantic-level resource slicing and bidirectional mapping to translate service requirements into communication and sensing resource allocation. Communication conveys source information available at the transmitter before transmission, whereas sensing acquires previously unknown environmental information through probing. However, these studies neither distinguish these information-generation mechanisms nor provide a general framework for jointly modeling and optimizing SemCom and SemS at the semantic level.


To address this gap, we build on prior SemCom studies \cite{xie2021deepsc,shao2022ib} and our recently developed SemS framework \cite{xiaoqi2026semanticSensing,zhang2026semantic} to develop semantic ISAC (SemISAC), a unified framework that advances ISAC from physical-resource sharing to semantic-level integration. SemISAC combines source semantics and sensing task information with available side information for shared waveform design and radio resource allocation, while the corresponding task decoders recover the source meaning and infer the required environmental information. Here, side information includes knowledge bases for semantic processing and available channel and resource information for physical design. The main contributions are summarized as follows.

\begin{itemize}

   \item We establish a general SemISAC framework that advances ISAC from physical-resource sharing to semantic-level integration through a unified information flow for source-semantic transmission and environmental-semantic acquisition. In particular, a joint semantic-to-physical (S2P) module combines source semantics and sensing task information with available side information to coordinate the shared waveform and radio resource allocation. We further provide an information-theoretic interpretation that distinguishes task-relevant information from physical information and characterizes how the shared physical design couples semantic communication and sensing, thereby revealing their semantic communication--sensing trade-off.

  \item Building on the proposed SemISAC framework, we formulate a general task-oriented optimization problem that jointly considers the semantic objectives of communication and sensing under shared physical-resource constraints. To support different implementation settings, we develop two realization methods, namely end-to-end (E2E) SemISAC optimization and modular SemISAC optimization. The former jointly optimizes the trainable functions throughout the processing chain using the final task losses. The latter optimizes selected modules using task losses evaluated through the processing chain or approximated by interfaces that map physical quality to task performance.

    \item As a concrete realization of modular SemISAC optimization, we develop a learning-based orthogonal frequency-division multiplexing (OFDM) design for representative SemCom and SemS tasks, where task-aware semantic proxies guide the joint design of a learnable time--frequency (TF) precoder. Specifically, the communication proxy weights errors in the recovered semantic representation by the local sensitivity of the classifier output. The sensing proxy weights delay--Doppler uncertainty by the local sensitivity of a smooth arrival-risk score. These proxies guide the precoder to adapt the shared physical resources to both semantic tasks. Simulation results demonstrate that the semantic proxies better reflect downstream semantic performance than the physical proxies. The proposed task-aware design achieves lower sensing semantic distortion under a given communication requirement and a more favorable communication--sensing trade-off than the baseline designs.

\end{itemize}

The remainder of this paper is organized as follows. Section II introduces SemCom and SemS and describes their corresponding information flows. Section III develops the unified SemISAC architecture and provides an information-theoretic interpretation of its communication--sensing trade-off. Section IV formulates the general task-oriented optimization problem and presents the E2E and modular realization methods. Section V presents the learning-based OFDM TF precoding realization. Section VI presents the simulation results, and Section VII concludes the paper.

\section{System Model}
\label{sec:system_model}

Conventional communication and sensing target source reconstruction and physical sensing results, respectively, while conventional ISAC coordinates a shared transmit waveform for these two physical objectives. SemCom and SemS instead focus on task-relevant information determined by their task specifications and knowledge bases while retaining distinct information-generation mechanisms. Accordingly, this section models the two information-generation processes separately.

\subsection{SemCom Model for Source Information Transmission}
\label{subsec:semantic_communication_model}

In SemCom, the current source realization is available at the transmitter, allowing the content required by a receiver-side task to be extracted and encoded before channel propagation. Let $\mathcal{T}_{\rm c}$ and $\mathcal{K}_{\rm c}$ denote the communication task specification and knowledge base, respectively. In a single-task implementation, $\mathcal{T}_{\rm c}$ need not be provided explicitly because the task can be implicitly encoded in the model architecture and loss function. In general, the source $\mathbf{S}_{\rm c}$ can comprise text, images, video, speech, or other forms of data. The associated communication task target is defined as \cite{shao2022ib}
\begin{align}
    \mathbf{U}_{\rm c}
    =
    \phi_{\rm c}
    \left(
        \mathbf{S}_{\rm c};
        \mathcal{T}_{\rm c},
        \mathcal{K}_{\rm c}
    \right),
    \label{eq:communication_task_target}
\end{align}
where $\phi_{\rm c}(\cdot)$ specifies the source information required by $\mathcal{T}_{\rm c}$ under $\mathcal{K}_{\rm c}$. The communication semantic encoder maps $\mathbf{S}_{\rm c}$ to the transmitter-side representation as
\begin{align}
    \mathbf{z}_{\rm c}
    =
    E_{\rm c}
    \left(
        \mathbf{S}_{\rm c},
        \mathcal{T}_{\rm c};
        \mathcal{K}_{\rm c},
        \boldsymbol{\Theta}_{\rm c}^{\rm e}
    \right),
    \label{eq:communication_semantic_representation}
\end{align}
where $E_{\rm c}(\cdot)$ denotes the communication semantic encoder and $\boldsymbol{\Theta}_{\rm c}^{\rm e}$ contains its trainable parameters. \revised{Transmit signal processing maps this representation to the communication signal $\mathbf{x}_{\rm c}$ through the physical operations required by the chosen transmission scheme. The representation $\mathbf{z}_{\rm c}$ serves as the functional interface between communication semantic encoding and transmit signal processing, whether these functions are implemented separately or jointly in an E2E design.} After channel propagation, the received communication signal is modeled as \cite{liu2022isac}
\begin{align}
    \mathbf{y}_{\rm c}
    =
    \mathbf{H}_{\rm c}\mathbf{x}_{\rm c}
    +
    \mathbf{n}_{\rm c},
    \label{eq:physical_communication_channel}
\end{align}
where $\mathbf{H}_{\rm c}$ denotes the physical communication channel, $\mathbf{n}_{\rm c}$ is the noise, and $\mathbf{y}_{\rm c}$ is the received communication signal.

\revised{At the receiver, communication receiver processing obtains the received representation from $\mathbf{y}_{\rm c}$, and the communication semantic decoder interprets it for the task. Their combined operation is represented as}
\begin{align}
    \widehat{\mathbf{U}}_{\rm c}
    =
    D_{\rm c}
    \left(
        \mathbf{y}_{\rm c};
        \mathcal{T}_{\rm c},
        \mathcal{K}_{\rm c},
        \boldsymbol{\Theta}_{\rm c}^{\rm d}
    \right),
    \label{eq:communication_task_output}
\end{align}
\revised{where $D_{\rm c}(\cdot)$ represents communication receiver processing and semantic decoding, and $\boldsymbol{\Theta}_{\rm c}^{\rm d}$ contains the receiver parameters.}

\subsection{SemS Model for Environmental Information Acquisition}
\label{subsec:semantic_sensing_model}

SemS acquires task-relevant environmental information by transmitting a waveform to probe the physical environment and observing the response through the sensing channel. Before sensing, the task specification and knowledge base are available, whereas the current environmental state \revised{$\mathbf{E}$} and the sensing-channel realization it induces remain unknown \cite{xiaoqi2026semanticSensing}. Therefore, the task specification and knowledge base guide waveform design to enhance sensitivity to the environmental features most informative for the sensing task. Let $\mathcal{T}_{\rm s}$ and $\mathcal{K}_{\rm s}$ denote the sensing task specification and knowledge base, respectively. The task-relevant information is formulated as
\begin{align}
    \mathbf{U}_{\rm s}
    =
    \phi_{\rm s}
    \left(
        \mathbf{E};
        \mathcal{T}_{\rm s},
        \mathcal{K}_{\rm s}
    \right),
    \label{eq:sensing_task_target}
\end{align}
where $\phi_{\rm s}(\cdot)$ maps the environmental state to the information required by $\mathcal{T}_{\rm s}$ under $\mathcal{K}_{\rm s}$. The same environment determines the propagation and reflection parameters $\mathcal{P}_{\rm phy}=\psi_{\rm phy}(\mathbf{E})$. Thus, $\mathbf{U}_{\rm s}$ specifies the task-relevant information to be inferred, whereas $\mathcal{P}_{\rm phy}$ determines how the environment affects the probing waveform. \revised{The task-relevant information can depend on physical-parameter values and their relationships, such as arrival time derived from distance and speed. The corresponding sensing semantic requirement is constructed as}
\begin{align}
    \mathbf{q}_{\rm s}
    =
    E_{\rm s}
    \left(
        \mathcal{T}_{\rm s};
        \mathcal{K}_{\rm s},
        \boldsymbol{\Theta}_{\rm s}^{\rm e}
    \right),
    \label{eq:sensing_task_representation}
\end{align}
\revised{where $E_{\rm s}(\cdot)$ constructs the requirement from the task specification and knowledge base, including transmitter-available prior information or context.} The parameters $\boldsymbol{\Theta}_{\rm s}^{\rm e}$ specify an analytical rule, a predefined configuration, or a learned mapping. \revised{The requirement $\mathbf{q}_{\rm s}$ describes the environmental information to be acquired, including relevant parameter ranges, relationships, and relative importance. It guides waveform design before the current environmental information is obtained through sensing. Similar to $\mathbf{z}_{\rm c}$, $\mathbf{q}_{\rm s}$ serves as the functional interface between sensing requirement construction and waveform design, whether these functions are implemented separately or jointly in an E2E design.}

The sensing waveform $\mathbf{x}_{\rm s}$ interacts with the environment through the physical sensing channel. The observation is modeled as \cite{zhang2022jcas}
\begin{align}
    \mathbf{y}_{\rm s}
    =
    \mathbf{H}_{\rm s}
    \left(
        \mathcal{P}_{\rm phy}
    \right)
    \mathbf{x}_{\rm s}
    +
    \mathbf{n}_{\rm s},
    \label{eq:physical_sensing_channel}
\end{align}
where $\mathbf{H}_{\rm s}(\mathcal{P}_{\rm phy})$ denotes the physical sensing channel governed by $\mathcal{P}_{\rm phy}$, $\mathbf{n}_{\rm s}$ denotes the channel noise, and $\mathbf{y}_{\rm s}$ denotes the physical sensing observation. \revised{At the receiver, sensing signal processing extracts sensing features from $\mathbf{y}_{\rm s}$ using the known waveform $\mathbf{x}_{\rm s}$, and the sensing semantic decoder interprets these features to infer the task-relevant environmental information.} Their combined operation is represented as
\begin{align}
    \widehat{\mathbf{U}}_{\rm s}
    =
    D_{\rm s}
    \left(
        \mathbf{y}_{\rm s},
        \mathbf{x}_{\rm s};
        \mathcal{T}_{\rm s},
        \mathcal{K}_{\rm s},
        \boldsymbol{\Theta}_{\rm s}^{\rm d}
    \right),
    \label{eq:sensing_task_output}
\end{align}
\revised{where $D_{\rm s}(\cdot)$ represents sensing signal processing and semantic decoding, and $\boldsymbol{\Theta}_{\rm s}^{\rm d}$ denotes the receiver parameters.}

SemCom and SemS retain distinct information-generation mechanisms but share a functional structure comprising task-dependent representations, physical waveforms, observations, and receiver-side task outputs. This common structure provides the basis for the unified SemISAC architecture developed in Section~\ref{sec:general_architecture}.

\section{General SemISAC Architecture}
\label{sec:general_architecture}

\begin{figure*}[tbh]
    \centering
    \includegraphics[width=0.85\linewidth]{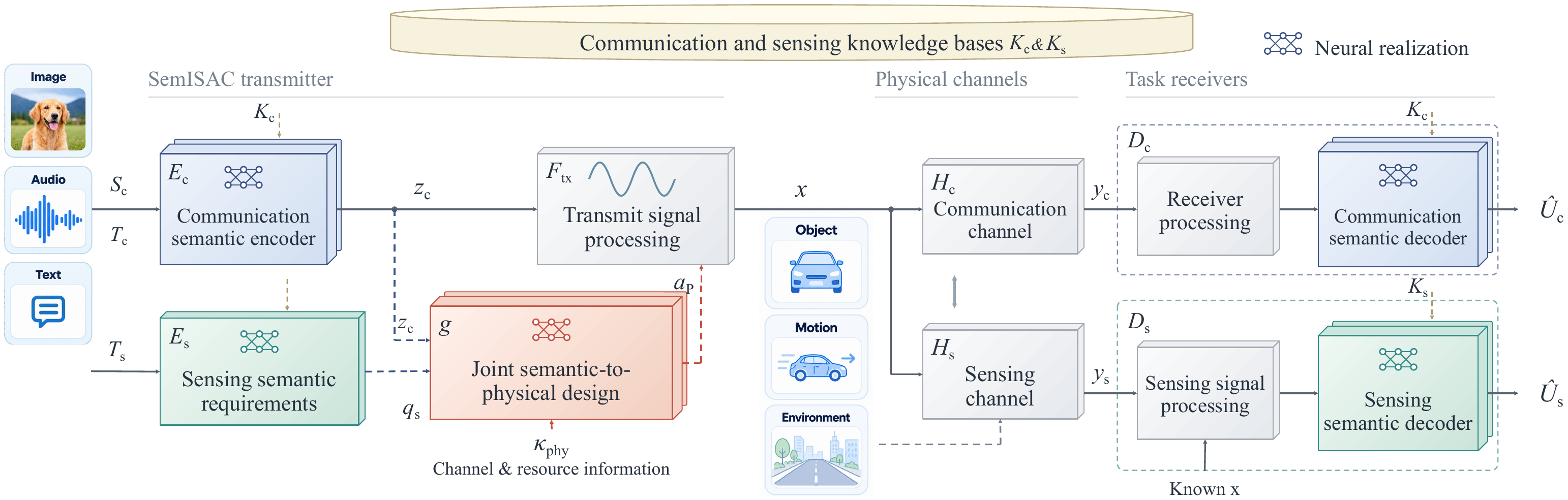}
    \caption{The proposed SemISAC architecture.}
    \label{fig:SS_frame}
\end{figure*}

Building on the branch-specific models in Section~\ref{sec:system_model}, we develop the SemISAC architecture from the signal and channel requirements of SemCom and SemS. We first introduce its functional components and information flow, and then provide an information-theoretic interpretation of the shared design.

\subsection{SemISAC Information Flow and Functional Architecture}
\label{subsec:core_components}

The SemISAC architecture in Fig.~\ref{fig:SS_frame} retains the shared physical transmission of ISAC and uses task-relevant information to guide its design. The communication task specifies what source information must reach the receiver, while the sensing task specifies what environmental information must be acquired. These requirements determine how the shared waveform and radio resources should be configured to support the two tasks. The neural-network symbols indicate functions that can be implemented using neural networks.

\revised{The communication and sensing knowledge bases determine task relevance and support semantic interpretation at the receivers. In particular, the communication knowledge base can be shared explicitly or embedded consistently in pretrained transmitter and receiver models \cite{yi2024sharedkb}.}

\begin{itemize}
    \item \emph{Communication semantic encoder.} \revised{Following SemCom, we use a task-dependent source representation to focus transmission resources on the information required by the receiver task. The encoder $E_{\rm c}$ uses $\mathcal{T}_{\rm c}$ and $\mathcal{K}_{\rm c}$ to extract task-relevant features from $\mathbf{S}_{\rm c}$, yielding $\mathbf{z}_{\rm c}$. This representation follows the signal path to transmit processing and also informs the physical design.}

    \item \emph{Sensing requirement construction.} To guide waveform design before probing, this module expresses the sensing task through the environmental features to be acquired and their relative importance. The requirement function $E_{\rm s}$ uses $\mathcal{T}_{\rm s}$ and $\mathcal{K}_{\rm s}$, including available sensing priors, to construct $\mathbf{q}_{\rm s}$.

    \item \emph{Joint S2P design.} \revised{We introduce joint S2P design to translate source semantics and sensing requirements into a shared transmit configuration, allowing semantic relevance to guide waveform design and resource allocation. This module uses the cross-layer inputs $(\mathbf{z}_{\rm c},\mathbf{q}_{\rm s})$ and the available channel, resource, and system information $\boldsymbol{\kappa}_{\rm phy}$ to select the transmit configuration $\mathbf{a}_{\rm P}$.}

    \item \emph{Transmit signal processing.} \revised{This module applies the selected configuration to the communication representation to generate the shared signal. It performs coding and symbol mapping, precoding and resource mapping, and waveform generation as required by the realization. In learned SemCom, the channel encoder maps semantic features to baseband channel-input symbols \cite{xie2021deepsc}. This function is contained in transmit signal processing and can be combined with precoding. A digital realization converts features to bits before forward error correction and symbol modulation, while a continuous-valued realization directly maps features to channel-input symbols. The communication path also supports joint source-channel coding (JSCC), in which source representation and channel-symbol generation are designed jointly.}

    \item \emph{Communication receive path.} \revised{Communication receiver processing provides a received representation for the communication semantic decoder to recover $\widehat{\mathbf{U}}_{\rm c}$.}

    \item \emph{Sensing receive path.} \revised{Sensing signal processing extracts sensing features from the echo using the known transmitted signal, and the sensing semantic decoder interprets these features to infer $\widehat{\mathbf{U}}_{\rm s}$.}
\end{itemize}

\revised{Each receive path can implement its two functions separately or jointly through the corresponding $D_i$ in Section~\ref{sec:system_model}.}

\revised{The joint S2P design and transmit signal processing are expressed as}
\begin{align}
    \begin{aligned}
    \mathbf{a}_{\rm P}
    ={}&
    g_{\boldsymbol{\omega}_{g}}
    \left(
        \mathbf{z}_{\rm c},
        \mathbf{q}_{\rm s};
        \boldsymbol{\kappa}_{\rm phy}
    \right), \\
    \mathbf{x}
    ={}&
    \revised{F_{\rm tx}
    \left(
        \mathbf{z}_{\rm c};
        \mathbf{a}_{\rm P},
        \boldsymbol{\omega}_{\rm tx}
    \right)},
    \end{aligned}
    \label{eq:joint_semisac_coordination}
\end{align}
\revised{where $\boldsymbol{\omega}_{g}$ parameterizes the design function and $\boldsymbol{\omega}_{\rm tx}$ denotes the optional trainable parameters of transmit signal processing. The configuration $\mathbf{a}_{\rm P}$ specifies the selected precoding, waveform settings, or resource allocation for the current transmission.}

\revised{The shared signal $\mathbf{x}$ replaces the separate waveforms in Eqs.~(\ref{eq:physical_communication_channel}) and~(\ref{eq:physical_sensing_channel}), producing the communication observation $\mathbf{y}_{\rm c}$ and sensing observation $\mathbf{y}_{\rm s}$. The two channels can coincide in bistatic sensing with a shared remote receiver and differ in monostatic sensing, where echoes return to \mbox{the transmitter~\cite{andrewradar2022}}.}

\subsection{Information-Theoretic Interpretation of SemISAC}
\label{subsec:information_theoretic_interpretation}

Information-theoretic ISAC measures the physical information supported by a shared waveform using communication mutual information and sensing-state estimation limits \cite{xiong2023fundamental}. This subsection distinguishes that physical information from the information required by the communication and sensing tasks. The resulting decomposition yields semantic performance bounds and identifies the task-sufficient case in which the physical and semantic trade-offs coincide.

Let $\boldsymbol{\Omega}=\left(\mathcal{T}_{\rm c},\mathcal{T}_{\rm s},\mathcal{K}_{\rm c},\mathcal{K}_{\rm s}\right)$ denote a given semantic context, and let $\Pi$ denote the set of feasible S2P policies under the waveform and resource constraints. \revised{For each $\pi\in\Pi$, the selected transmit configuration and transmit signal processing determine $X$.} All mutual-information quantities use the joint distribution induced by $\pi$ for this context.

The variables $X$, $Y_{\rm c}$, $Y_{\rm s}$, $U_{\rm c}$, $U_{\rm s}$, $E$, and $P_{\rm phy}$ correspond to $\mathbf{x}$, $\mathbf{y}_{\rm c}$, $\mathbf{y}_{\rm s}$, $\mathbf{U}_{\rm c}$, $\mathbf{U}_{\rm s}$, $\mathbf{E}$, and $\mathcal{P}_{\rm phy}$. The semantic information in the two branches is
\begin{align*}
    \mathcal{I}_{\rm c}^{\rm sem}(\pi)
    &\triangleq I_{\pi}(U_{\rm c};Y_{\rm c}),
    &\qquad \mathcal{I}_{\rm s}^{\rm sem}(\pi)
    &\triangleq I_{\pi}(U_{\rm s};Y_{\rm s}\mid X).
\end{align*}
Because the probing waveform $X$ is known at the sensing receiver, $\mathcal{I}_{\rm s}^{\rm sem}(\pi)$ is conditioned on $X$. The physical information in the two branches is
\begin{align*}
    \mathcal{I}_{\rm c}^{\rm phy}(\pi)
    &\triangleq I_{\pi}(X;Y_{\rm c}),
    &\qquad \mathcal{I}_{\rm s}^{\rm phy}(\pi)
    &\triangleq I_{\pi}(P_{\rm phy};Y_{\rm s}\mid X).
\end{align*}

For branch $i\in\{{\rm c},{\rm s}\}$, $\mathcal{L}_{i}^{\rm sem}(\pi)\triangleq\mathbb{E}_{\pi}[d_i(U_i,\widehat{U}_i)]$ is the achieved task distortion. Under the $\pi$-induced distribution and distortion measures $d_i$, the rate--distortion quantities are $R_{U_{\rm c}}(\mathcal{L}_{\rm c}^{\rm sem}(\pi))$ and $R_{U_{\rm s}\mid X}(\mathcal{L}_{\rm s}^{\rm sem}(\pi))$ \cite{cover2006elementsInformation,stavrou2022goal}. The sensing model in \eqref{eq:physical_sensing_channel} satisfies
\begin{align}
    p_{\pi}
    \left(
        y_{\rm s}
        \mid
        e,p_{\rm phy},x
    \right)
    &=
    p_{\pi}
    \left(
        y_{\rm s}
        \mid
        p_{\rm phy},x
    \right).
    \label{eq:sensing_conditional_observation_assumption}
\end{align}
Since $U_{\rm s}=\phi_{\rm s}(E;\mathcal{T}_{\rm s},\mathcal{K}_{\rm s})$ and $P_{\rm phy}=\psi_{\rm phy}(E)$, \eqref{eq:sensing_conditional_observation_assumption} implies $I_{\pi}(U_{\rm s};Y_{\rm s}\mid P_{\rm phy},X)=0$. Equivalently, $U_{\rm s}\longrightarrow P_{\rm phy}\longrightarrow Y_{\rm s}$ is a Markov chain conditioned on $X$.

\begin{proposition}[Task-Relevant Information Decomposition]
\label{prop:task_relevant_information_decomposition}

Under the physical observation models in Section~\ref{sec:system_model}, suppose that the distribution induced by $\pi\in\Pi$ satisfies $U_{\rm c}\longrightarrow X\longrightarrow Y_{\rm c}$. This communication Markov chain and the sensing relation $I_{\pi}(U_{\rm s};Y_{\rm s}\mid P_{\rm phy},X)=0$ give the following decompositions.
\begin{align}
    \mathcal{I}_{\rm c}^{\rm phy}(\pi)
    &=
    \mathcal{I}_{\rm c}^{\rm sem}(\pi)
    +
    I_{\pi}
    \left(
        X;
        Y_{\rm c}
        \mid
        U_{\rm c}
    \right),
    \label{eq:communication_information_decomposition}\\
    \mathcal{I}_{\rm s}^{\rm phy}(\pi)
    &=
    \mathcal{I}_{\rm s}^{\rm sem}(\pi)
    +
    I_{\pi}
    \left(
        P_{\rm phy};
        Y_{\rm s}
        \mid
        U_{\rm s},X
    \right).
    \label{eq:sensing_information_decomposition}
\end{align}
The two conditional mutual-information terms are $\Delta_{\rm c}(\pi)$ and $\Delta_{\rm s}(\pi)$. Their non-negativity gives $\mathcal{I}_{i}^{\rm sem}(\pi)\leq\mathcal{I}_{i}^{\rm phy}(\pi)$ for $i\in\{{\rm c},{\rm s}\}$.
\end{proposition}

\begin{IEEEproof}
The chain $U_{\rm c}\longrightarrow X\longrightarrow Y_{\rm c}$ implies $I_{\pi}(U_{\rm c};Y_{\rm c}\mid X)=0$. The sensing relation gives $I_{\pi}(U_{\rm s};Y_{\rm s}\mid P_{\rm phy},X)=0$. The chain rule then gives
\begin{align*}
    I_{\pi}(X;Y_{\rm c})
    &=I_{\pi}(U_{\rm c},X;Y_{\rm c})\\
    &=I_{\pi}(U_{\rm c};Y_{\rm c})
    +I_{\pi}(X;Y_{\rm c}\mid U_{\rm c}),\\
    I_{\pi}(P_{\rm phy};Y_{\rm s}\mid X)
    &=I_{\pi}(U_{\rm s},P_{\rm phy};Y_{\rm s}\mid X)\\
    &=I_{\pi}(U_{\rm s};Y_{\rm s}\mid X)
    +I_{\pi}(P_{\rm phy};Y_{\rm s}\mid U_{\rm s},X).
\end{align*}
Substituting the information definitions gives \eqref{eq:communication_information_decomposition} and \eqref{eq:sensing_information_decomposition}. Non-negativity of conditional mutual information gives the stated inequalities.
\end{IEEEproof}

\begin{corollary}[Semantic Performance Bounds]
\label{cor:semantic_performance_bounds}
Under the conditions of Proposition~\ref{prop:task_relevant_information_decomposition}, the achieved task distortions satisfy
\begin{align}
    R_{U_{\rm c}}\left(\mathcal{L}_{\rm c}^{\rm sem}(\pi)\right)
    &\leq I_{\pi}\left(U_{\rm c};\widehat{U}_{\rm c}\right)
    \leq \mathcal{I}_{\rm c}^{\rm sem}(\pi)
    \leq \mathcal{I}_{\rm c}^{\rm phy}(\pi),
    \label{eq:communication_semantic_performance_bound}\\
    R_{U_{\rm s}\mid X}&\left(\mathcal{L}_{\rm s}^{\rm sem}(\pi)\right)
    \leq I_{\pi}\left(U_{\rm s};\widehat{U}_{\rm s}\mid X\right)
    \leq \mathcal{I}_{\rm s}^{\rm sem}(\pi)
    \nonumber\\
    &\leq \min\left\{I_{\pi}\left(U_{\rm s};P_{\rm phy}\mid X\right),\mathcal{I}_{\rm s}^{\rm phy}(\pi)\right\}.
    \label{eq:sensing_semantic_performance_bound}
\end{align}
\end{corollary}

\begin{IEEEproof}
The rate--distortion relations give the first inequality in each bound. Since $\widehat{U}_{\rm c}$ is generated from $Y_{\rm c}$ and $\widehat{U}_{\rm s}$ is generated from $(Y_{\rm s},X)$, the data-processing inequality gives the second inequality. Proposition~\ref{prop:task_relevant_information_decomposition} gives the physical information upper bounds. The sensing Markov chain also gives $\mathcal{I}_{\rm s}^{\rm sem}(\pi)\leq I_{\pi}(U_{\rm s};P_{\rm phy}\mid X)$.
\end{IEEEproof}

The communication residual $\Delta_{\rm c}(\pi)$ is the information that $Y_{\rm c}$ contains about the transmitted physical signal beyond $U_{\rm c}$. The sensing residual $\Delta_{\rm s}(\pi)$ is the information that $Y_{\rm s}$ contains about the physical environmental parameters beyond $U_{\rm s}$ when $X$ is known. Increasing information contained only in a residual does not increase $\mathcal{I}_{i}^{\rm sem}(\pi)$.

The bounds in \eqref{eq:communication_semantic_performance_bound} and \eqref{eq:sensing_semantic_performance_bound} compare the information required for the achieved task distortion with the task-relevant information available in each observation and retained by the decoded output. For sensing, $I_{\pi}(U_{\rm s};P_{\rm phy}\mid X)$ represents the task-relevant information contained in the physical environmental parameters $P_{\rm phy}$ when $X$ is known. The quantity $\mathcal{I}_{\rm s}^{\rm phy}(\pi)$ represents the information that the sensing observation $Y_{\rm s}$ provides about these parameters under the same condition. Semantic sensing therefore depends on both the task relevance of these parameters and the information about them acquired through waveform--environment interaction.

For $j\in\{{\rm sem},{\rm phy}\}$, let $\boldsymbol{\mathcal{I}}^{j}(\pi)=[\mathcal{I}_{\rm c}^{j}(\pi),\mathcal{I}_{\rm s}^{j}(\pi)]^{T}$ and $\boldsymbol{\Delta}(\pi)=[\Delta_{\rm c}(\pi),\Delta_{\rm s}(\pi)]^{T}$. Proposition~\ref{prop:task_relevant_information_decomposition} gives the policy-wise relation
\begin{align}
    \boldsymbol{\mathcal{I}}^{\rm sem}(\pi)
    &=\boldsymbol{\mathcal{I}}^{\rm phy}(\pi)
    -\boldsymbol{\Delta}(\pi),
    \qquad \boldsymbol{\Delta}(\pi)\succeq\mathbf{0}.
    \label{eq:policywise_information_relation}
\end{align}
For a given policy, both residuals vanish if and only if
\begin{align}
    Y_{\rm c}
    &\perp\!\!\!\perp X\mid U_{\rm c},
    &\qquad Y_{\rm s}
    &\perp\!\!\!\perp P_{\rm phy}\mid (U_{\rm s},X).
    \label{eq:zero_residual_conditions}
\end{align}
These conditional-independence relations define task-sufficient observations. In this case, $Y_{\rm c}$ contains no information about $X$ beyond $U_{\rm c}$, and $Y_{\rm s}$ contains no information about $P_{\rm phy}$ beyond $U_{\rm s}$ when $X$ is known.

For $j\in\{{\rm sem},{\rm phy}\}$, $\mathcal{R}_{j}$ denotes the semantic information region for $j={\rm sem}$ and the physical information region for $j={\rm phy}$, defined as
\begin{align}
    \mathcal{R}_{j}
    &\triangleq
    \operatorname{cl}
    \left\{
        \left(
            \mathcal{I}_{\rm c}^{j}(\pi),
            \mathcal{I}_{\rm s}^{j}(\pi)
        \right)
        \,\middle|\,
        \pi\in\Pi
    \right\}.
    \label{eq:semantic_information_region}
\end{align}
The Pareto boundaries of the semantic information region $\mathcal{R}_{\rm sem}$ and physical information region $\mathcal{R}_{\rm phy}$ characterize the task-relevant and physical information trade-offs, respectively, under the same feasible policy set.

\begin{corollary}[Semantic and Physical Boundary Equivalence]
\label{cor:semantic_physical_boundary_equivalence}
If $\Delta_{\rm c}(\pi)=\Delta_{\rm s}(\pi)=0$ for every $\pi\in\Pi$, then $\mathcal{R}_{\rm sem}=\mathcal{R}_{\rm phy}$ and their Pareto boundaries coincide.
\end{corollary}

\begin{IEEEproof}
Under the zero-residual condition, \eqref{eq:policywise_information_relation} gives $\boldsymbol{\mathcal{I}}^{\rm sem}(\pi)=\boldsymbol{\mathcal{I}}^{\rm phy}(\pi)$ for every feasible policy. Taking the closure over the same set $\Pi$ gives $\mathcal{R}_{\rm sem}=\mathcal{R}_{\rm phy}$. Their Pareto boundaries therefore coincide.
\end{IEEEproof}

Equation~\eqref{eq:policywise_information_relation} shows that the semantic information coordinates depend on both physical information and the policy-dependent residuals. Outside the task-sufficient case, these residuals can reorder feasible policies in the semantic information region. Physical Pareto optimality therefore does not generally imply semantic Pareto optimality.

The same \revised{joint S2P design policy} determines the shared physical signal, which governs both observations and couples $\mathcal{I}_{\rm c}^{\rm sem}(\pi)$ and $\mathcal{I}_{\rm s}^{\rm sem}(\pi)$ through the physical subsystem. The performance bounds in Corollary~\ref{cor:semantic_performance_bounds} constrain the achievable task losses according to the available task-relevant information. Section~\ref{sec:optimization_implementation} adopts these losses as the communication and sensing design objectives under the shared physical-resource constraints.

\section{Task-Oriented Optimization and Realization Framework}
\label{sec:optimization_implementation}

In this section, we formulate the general SemISAC design problem. We then present E2E SemISAC optimization and modular SemISAC optimization as two solution approaches under the same semantic objectives. Finally, we describe their offline design procedures and common online execution architecture.

\subsection{General SemISAC Design Formulation}
\label{subsec:semantic_task_metrics}
\label{subsec:general_semantic_optimization}

\subsubsection{Semantic Task Objectives}

Each branch is evaluated by the semantic task loss $\mathcal{L}_{i}^{\rm sem}(\pi)$. We suppress its explicit dependence on $\pi$ in this section and write $\mathcal{L}_{i}^{\rm sem}$ for $i\in\{{\rm c},{\rm s}\}$. The task-dependent function $d_i(\cdot,\cdot)$ may represent reconstruction, feature, or inference loss for communication and sensing. Unlike physical metrics such as rate, signal-to-interference-plus-noise ratio, estimation error, and the CRB, these losses directly evaluate task effectiveness \cite{zhang2022jcas,wen2024taskiscc}. Physical metrics remain useful as tractable objectives or constraints, but their relationship to semantic performance depends on the task and receiver.

\subsubsection{Constrained and Weighted Design Formulations}

The general SemISAC design formulation jointly optimizes the communication and sensing semantic objectives under shared physical resources. \revised{Let $\boldsymbol{\omega}$ denote the parameters of the semantic functions, joint S2P design, transmit signal processing, and receive paths.} When sensing is the primary task, the corresponding formulation is given by
\begin{align}
    \begin{aligned}
    \min_{\boldsymbol{\omega}} \quad &\mathcal{L}_{\rm s}^{\rm sem},
    \quad \mathrm{s.t.} \quad
    \mathcal{L}_{\rm c}^{\rm sem}\leq\epsilon_{\rm c},
    \quad \mathbf{x}\in\mathcal{X}(P_{\max}),\\[-0.2em]
    &\mathcal{C}_{j}\left(\mathbf{x},\mathbf{a}_{\rm P},\boldsymbol{\omega}\right)\leq 0,
    \quad j=1,\ldots,J.
    \end{aligned}
    \label{eq:general_semisac_optimization}
\end{align}
Here, $\epsilon_{\rm c}$ denotes the required communication semantic-loss level, $\mathcal{X}(P_{\max})$ denotes the feasible waveform set under the power budget $P_{\max}$, and $\mathcal{C}_{j}(\mathbf{x},\mathbf{a}_{\rm P},\boldsymbol{\omega})$ represents the $j$th physical-design constraint. The task roles can be exchanged when communication is the primary objective. When neither task has a prescribed service constraint, the design minimizes a weighted sum of the two losses as
\begin{align}
    \min_{\boldsymbol{\omega}}
    \quad
    \beta
    \mathcal{L}_{\rm c}^{\rm sem}
    +
    (1-\beta)
    \mathcal{L}_{\rm s}^{\rm sem},
    \label{eq:weighted_semisac_optimization}
\end{align}
where $\beta\in[0,1]$ determines the communication--sensing task priority. To address the general SemISAC problem, we propose two optimization methods similar to those used in SemCom \cite{cai2025e2e}, namely E2E SemISAC optimization and modular SemISAC optimization.

\subsection{E2E SemISAC Optimization}
\label{subsec:end_to_end_semisac_optimization}

E2E SemISAC optimization evaluates the final communication and sensing semantic losses through the complete processing chain. \revised{The chain includes communication semantic encoding, sensing requirement construction, joint S2P design, transmit signal processing, physical interaction, and the receive paths. When it is differentiable, task-loss gradients propagate through these functions to update their trainable parameters.} For nondifferentiable chains or task metrics, policy-based or other non-gradient methods can optimize the same final semantic objectives \cite{getu2025semantic}.

Let $\mathcal{D}_{\rm tr}=\{(\boldsymbol{\chi}^{(n)},\mathbf{U}_{\rm c}^{(n)},\mathbf{U}_{\rm s}^{(n)})\}_{n=1}^{N_{\rm tr}}$ denote a representative offline training set, where $\boldsymbol{\chi}^{(n)}$ represents the information supplied to the trainable modules. For E2E SemISAC optimization, the full parameter set is defined as
\begin{align}
    \begin{aligned}
    \boldsymbol{\omega}
    \triangleq
    \big\{
        \boldsymbol{\Theta}_{\rm c}^{\rm e},
        \boldsymbol{\Theta}_{\rm s}^{\rm e},
        \boldsymbol{\omega}_{g},
        \revised{\boldsymbol{\omega}_{\rm tx},}
        \boldsymbol{\Theta}_{\rm c}^{\rm d},
        \boldsymbol{\Theta}_{\rm s}^{\rm d}
    \big\}.
    \end{aligned}
    \label{eq:semisac_parameter_set}
\end{align}
\revised{The six groups correspond to communication semantic encoding, sensing requirement construction, joint S2P design, transmit signal processing, and the communication and sensing receive paths, respectively.} The optimization updates the trainable entries, while the remaining parameters retain their prescribed values. In particular, $\boldsymbol{\omega}_{\rm tx}$ is empty when transmit signal processing has no independent trainable parameters, and a predefined sensing requirement rule keeps $\boldsymbol{\Theta}_{\rm s}^{\rm e}$ unchanged. \revised{E2E SemISAC optimization updates the trainable groups through the empirical problem}
\begin{align}
    \boldsymbol{\omega}^{\star}
    =
    \arg\min_{\boldsymbol{\omega}}
    \quad &
    \beta
    \widehat{\mathcal{L}}_{\rm c}^{\rm sem}
    \left(\boldsymbol{\omega}\right)
    +
    (1-\beta)
    \widehat{\mathcal{L}}_{\rm s}^{\rm sem}
    \left(\boldsymbol{\omega}\right)
    \nonumber\\[-0.3em] 
    &+
    \lambda_{\rm p}
    \mathcal{R}_{\rm p}
    \left(\boldsymbol{\omega}\right)
    +
    \sum_{j=1}^{J}
    \lambda_j
    \left[
        \widehat{\mathcal{C}}_j
        \left(\boldsymbol{\omega}\right)
    \right]_{+},
    \label{eq:joint_end_to_end_training}
\end{align}
where $\widehat{\mathcal{L}}_{i}^{\rm sem}$ is the empirical semantic task loss over $\mathcal{D}_{\rm tr}$ and $\mathcal{R}_{\rm p}$ is a penalty term for waveform or physical-resource requirements. The quantity $\widehat{\mathcal{C}}_j$ is the empirical value of constraint function $\mathcal{C}_j$, and $[a]_{+}=\max(a,0)$ extracts the positive part of $a$. The coefficients $\lambda_{\rm p}$ and $\lambda_j$ weight the physical-resource and constraint-violation penalties, respectively. A differentiable processing chain supports gradient-based joint training and, when data cannot be centralized, distributed training \cite{cai2025e2e,liu2023ambient}. This implementation generally requires representative training conditions and access to the complete chain during optimization.

\subsection{Modular SemISAC Optimization}
\label{subsec:modular_semantic_optimization}

Modular SemISAC optimization preserves the general semantic objectives while optimizing a selected parameter subset and holding the remaining parameters fixed for the chosen system realization. Modules in the fixed subset may be independently developed, pretrained, or nondifferentiable \cite{xie2021deepsc}. For a given partition, inter-module interfaces specify how the optimized modules affect the final semantic task losses through the processing chain \cite{ren2025semantic_empowered_isac}.

Using the aggregate parameter set $\boldsymbol{\omega}$, a modular realization defines the partition as
\begin{align}
    \boldsymbol{\omega}
    =
    \left\{
        \boldsymbol{\omega}_{\rm o},
        \boldsymbol{\omega}_{\rm f}
    \right\},
    \label{eq:modular_parameter_partition}
\end{align}
where $\boldsymbol{\omega}_{\rm o}$ and $\boldsymbol{\omega}_{\rm f}$ denote the optimized and fixed parameter subsets, respectively. Conditioned on $\boldsymbol{\omega}_{\rm f}$, the modular formulation retains the E2E objective structure given by
\begin{align}
    \boldsymbol{\omega}_{\rm o}^{\star}
    =
    \arg &\min_{\boldsymbol{\omega}_{\rm o}}
    \quad 
    \beta
    \widehat{\mathcal{L}}_{\rm c}^{\rm sem}
    \left(
        \boldsymbol{\omega}_{\rm o};
        \boldsymbol{\omega}_{\rm f}
    \right)
    +
    (1-\beta)
    \widehat{\mathcal{L}}_{\rm s}^{\rm sem}
    \left(
        \boldsymbol{\omega}_{\rm o};
        \boldsymbol{\omega}_{\rm f}
    \right)
    \nonumber\\
    &+
    \lambda_{\rm p}
    \mathcal{R}_{\rm p}
    \left(
        \boldsymbol{\omega}_{\rm o};
        \boldsymbol{\omega}_{\rm f}
    \right)
    +
    \sum_{j=1}^{J}
    \lambda_j
    \left[
        \widehat{\mathcal{C}}_j
        \left(
            \boldsymbol{\omega}_{\rm o};
            \boldsymbol{\omega}_{\rm f}
        \right)
    \right]_{+},
    \label{eq:modular_transmitter_training}
\end{align}
E2E SemISAC optimization in \eqref{eq:joint_end_to_end_training} jointly updates \revised{all trainable entries in $\boldsymbol{\omega}$}, whereas modular SemISAC optimization in \eqref{eq:modular_transmitter_training} updates $\boldsymbol{\omega}_{\rm o}$ with $\boldsymbol{\omega}_{\rm f}$ fixed. The modular loss $\widehat{\mathcal{L}}_{i}^{\rm sem}$ can be evaluated by propagating the optimized modules' outputs through the downstream processing chain. When repeated execution is costly or impractical, particularly for complex or nondifferentiable modules, a calibrated interface can approximate the relationship between the optimized modules' physical output and the final semantic loss \cite{he2026sensem}.

For task $i\in\{{\rm c},{\rm s}\}$, let $\mathbf{q}_{i}^{\rm phy}$ denote the physical-quality variables determined by the selected physical design and relevant to downstream task processing. Evaluating these downstream modules under representative physical conditions gives the calibrated approximation
\begin{align}
    \widehat{\mathcal{L}}_{i}^{\rm sem}
    \approx
    \widehat{\Gamma}_{i}
    \left(
        \mathbf{q}_{i}^{\rm phy};
        \boldsymbol{\omega}_{\rm f}
    \right),
    \quad
    i\in\{{\rm c},{\rm s}\},
    \label{eq:semantic_to_physical_mapping}
\end{align}
where $\widehat{\Gamma}_{i}(\cdot)$ is an analytically derived or empirically calibrated mapping from physical quality to an estimated semantic loss for the given module parameters $\boldsymbol{\omega}_{\rm f}$ \cite{wen2024taskiscc}.

\subsection{Offline Design and Online Execution}
\label{subsec:offline_online_framework}

During offline development, the E2E and modular SemISAC optimization methods determine their respective deployable parameter sets using representative source, environmental, channel, and task data. These offline data include complete environmental states and sensing targets for evaluating the semantic objectives.

Online, both methods use the same architecture with the resulting parameter sets $\widetilde{\boldsymbol{\omega}}=\boldsymbol{\omega}^{\star}$ for E2E SemISAC optimization and $\widetilde{\boldsymbol{\omega}}=\{\boldsymbol{\omega}_{\rm o}^{\star},\boldsymbol{\omega}_{\rm f}\}$ for modular SemISAC optimization. At each interval, $E_{\rm c}$ maps $\mathbf{S}_{\rm c}$ to $\mathbf{z}_{\rm c}$ using $\widetilde{\boldsymbol{\Theta}}_{\rm c}^{\rm e}$, while \revised{$E_{\rm s}$ constructs $\mathbf{q}_{\rm s}$ from the sensing task and available prior information using $\widetilde{\boldsymbol{\Theta}}_{\rm s}^{\rm e}$} according to Eqs.~(\ref{eq:communication_semantic_representation}) and~(\ref{eq:sensing_task_representation}). \revised{Joint S2P design selects the current transmit configuration, and transmit signal processing generates the shared signal as}
\begin{align}
    \begin{aligned}
    \mathbf{a}_{\rm P}
    ={}&
    g_{\widetilde{\boldsymbol{\omega}}_{g}}\left(
        \mathbf{z}_{\rm c},\mathbf{q}_{\rm s};
        \boldsymbol{\kappa}_{\rm phy}
    \right), \\
    \mathbf{x}
    ={}&
    \revised{F_{\rm tx}\left(
        \mathbf{z}_{\rm c};
        \mathbf{a}_{\rm P},
        \widetilde{\boldsymbol{\omega}}_{\rm tx}
    \right)}.
    \end{aligned}
    \label{eq:general_online_waveform_generation}
\end{align}
The shared signal produces $\mathbf{y}_{\rm c}$ and $\mathbf{y}_{\rm s}$ through Eqs.~(\ref{eq:physical_communication_channel}) and~(\ref{eq:physical_sensing_channel}). The receive functions $D_{\rm c}$ and $D_{\rm s}$ then map $\mathbf{y}_{\rm c}$ and $(\mathbf{y}_{\rm s},\mathbf{x})$ to $\widehat{\mathbf{U}}_{\rm c}$ and $\widehat{\mathbf{U}}_{\rm s}$, respectively, according to Eqs.~(\ref{eq:communication_task_output}) and~(\ref{eq:sensing_task_output}).


\section{Modular SemISAC Realization via Learning-Based OFDM TF Precoding}
\label{sec:ofdm_case_study}
\label{sec:deep_learning_sequence_precoder}

This section applies modular SemISAC optimization from Section~\ref{subsec:modular_semantic_optimization} to the OFDM realization in Fig.~\ref{fig:ofdm_realization}. We first specify the representative tasks and shared OFDM signal model, then formulate TF-precoder optimization using the modular \revised{task-loss interfaces}. In addition, we develop a learning-based S2P policy for low-complexity online execution.

\begin{figure}[t]
    \centering
    \includegraphics[width=\linewidth]{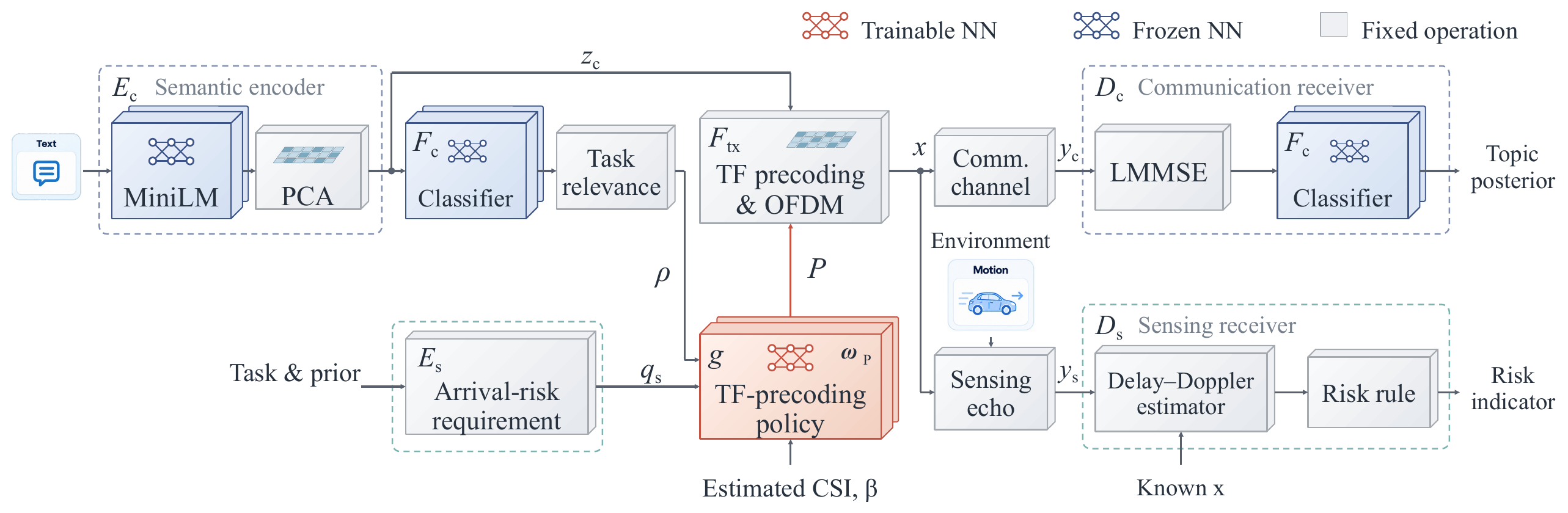}
    \caption{Modular OFDM realization of the proposed SemISAC framework.}
    \label{fig:ofdm_realization}
\end{figure}

\subsection{Communication and Sensing Task Instantiation}

\revised{The communication branch considers receiver-side topic inference for AG News text. The transmitter conveys a latent representation of each text sample, and the receiver uses the recovered representation to reproduce the reference posterior over four topic categories.} At frame $t$, the communication semantic encoder maps the text source to $\mathbf{z}_{{\rm c},t}=E_{\rm c}(\mathbf{S}_{{\rm c},t})\in\mathbb{C}^{Q}$. The implementation centers an all-MiniLM-L6-v2 sentence embedding \cite{allMiniLML6v2ModelCard,reimers2019sentenceBERT}, applies whitening based on principal component analysis \cite{kessy2018optimalWhitening} to obtain $2Q$ real dimensions, pairs them into $Q$ complex components, and applies normalization. The PCA block in Fig.~\ref{fig:ofdm_realization} represents these preprocessing operations applied to the MiniLM embedding.

Let $F_{\rm c}(\cdot)$ denote the frozen four-class classifier used at both the transmitter and receiver, and let $\widehat{\mathbf{z}}_{{\rm c},t}$ denote the recovered representation. The communication receive path $D_{\rm c}$ maps the received signal $\mathbf{y}_{{\rm c},t}$ to the topic posterior $\widehat{\mathbf{p}}_t$. The communication task target, receiver output, and Kullback--Leibler (KL) sample loss are
\begin{align}
    \begin{aligned}
    &\mathbf{U}_{{\rm c},t}
    \equiv
    \mathbf{p}_t
    =
    F_{\rm c}\left(\mathbf{z}_{{\rm c},t}\right),
    \quad
    \widehat{\mathbf{U}}_{{\rm c},t}
    \equiv
    \widehat{\mathbf{p}}_t
    =
    F_{\rm c}\left(\widehat{\mathbf{z}}_{{\rm c},t}\right),
    \\
    &d_{\rm c}
    \left(
        \mathbf{U}_{{\rm c},t},
        \widehat{\mathbf{U}}_{{\rm c},t}
    \right)
    =
    D_{\rm KL}
    \left(
        \mathbf{p}_t
        \,\Vert\,
        \widehat{\mathbf{p}}_t
    \right).
    \end{aligned}
    \label{eq:case_communication_task_instantiation}
\end{align}
The communication semantic loss is $\mathcal{L}_{\rm c}^{\rm sem}=\mathbb{E}[D_{\rm KL}(\mathbf{p}_t\Vert\widehat{\mathbf{p}}_t)]$. Let $\boldsymbol{\rho}_t$ denote the normalized componentwise task relevance of $\mathbf{z}_{{\rm c},t}$ derived from the local sensitivity of the classifier $F_{\rm c}$. \revised{It guides the protection of the latent components, while $\mathbf{z}_{{\rm c},t}$ carries the communication content through the physical channel.} The construction of $\boldsymbol{\rho}_t$ is given in Appendix~\ref{app:ofdm_proxy_construction}.

The sensing branch determines whether an approaching target reaches a given boundary within the safety horizon $T_{\rm safe}$. Let $\boldsymbol{\theta}_t=[\theta_{\tau,t},\theta_{\nu,t}]^T$ denote the unknown current normalized delay--Doppler state. The physical delay and Doppler shift are $\tau_t^{\rm phy}=\tau_0+s_{\tau}\theta_{\tau,t}$ and $\nu_t^{\rm phy}=\nu_0+s_{\nu}\theta_{\nu,t}$, where $\tau_0$ and $\nu_0$ are reference values and $s_{\tau}$ and $s_{\nu}$ are normalization scales. In the monostatic setting considered here, an approaching target with $\nu_t^{\rm phy}<0$ has distance and speed $d_t=c\tau_t^{\rm phy}/2$ and $v_t=-c\nu_t^{\rm phy}/(2f_c)>0$, where $c$ is the speed of light and $f_c$ is the carrier frequency. For a boundary at distance $d_b$, the arrival time is
\begin{align}
    t_{\rm arr}
    \left(
        \boldsymbol{\theta}_t
    \right)
    =
    \frac{d_t-d_b}{v_t}
    =
    -
    \frac{
        f_c
        \left(
            \tau_t^{\rm phy}
            -
            2d_b/c
        \right)
    }{
        \nu_t^{\rm phy}
    }.
    \label{eq:case_sensing_arrival_time_mapping}
\end{align}
Accordingly, the sensing task target is $U_{{\rm s},t}=\mathbbm{1}\left[t_{\rm arr}\left(\boldsymbol{\theta}_t\right)\leq T_{\rm safe}\right]$. For this simple arrival-risk task, the explicit mapping from distance and speed to the required decision enables analytical sensing requirement construction and semantic decoding.

Before probing, the analytical requirement function $E_{\rm s}$ constructs $\mathbf{q}_{{\rm s},t}=E_{\rm s}(\mathcal{T}_{\rm s};\mathcal{K}_{\rm s},\boldsymbol{\Theta}_{\rm s}^{\rm e})\equiv\operatorname{vech}(\mathbf{W}_{{\rm s},t})$ from the transmitter-available prior state $\overline{\boldsymbol{\theta}}_t$ and arrival-risk task. The normalized, symmetric matrix $\mathbf{W}_{{\rm s},t}\in\mathbb{R}^{2\times2}$ is formed from the outer product of the gradient of a smooth arrival-risk score at this prior, as derived in Appendix~\ref{app:ofdm_proxy_construction}. The sensing semantics enter the shared waveform design through $\mathbf{W}_{{\rm s},t}$, which weights local delay--Doppler error directions according to their relevance to the arrival-risk decision. The operator $\operatorname{vech}(\cdot)$ stacks the lower-triangular entries of a symmetric matrix column by column into a vector. In $D_{\rm s}$, sensing signal processing estimates delay and Doppler from $(\mathbf{y}_{{\rm s},t},\mathbf{x}_t)$ and converts them to the normalized state estimate $\widehat{\boldsymbol{\theta}}_t$. The sensing semantic decoder then compares the arrival time computed from this estimate with $T_{\rm safe}$. The resulting sensing task output is
\begin{align}
    \widehat{U}_{{\rm s},t}
    =
    D_{\rm s}
    \bigl(
        \mathbf{y}_{{\rm s},t},
        \mathbf{x}_t;
        \mathcal{T}_{\rm s},
        \mathcal{K}_{\rm s},
        \boldsymbol{\Theta}_{\rm s}^{\rm d}
    \bigr)
    =
    \mathbbm{1}
    \bigl[
        t_{\rm arr}
        \bigl(
            \widehat{\boldsymbol{\theta}}_t
        \bigr)
        \leq
        T_{\rm safe}
    \bigr].
    \label{eq:case_sensing_task_output}
\end{align}
The resulting semantic loss is
$\mathcal{L}_{\rm s}^{\rm sem}
=
\Pr(\widehat{U}_{{\rm s},t}\neq U_{{\rm s},t})$,
with the delay--Doppler estimator detailed in Appendix~\ref{app:ofdm_proxy_construction}.

\subsection{Shared OFDM Signal}

We consider a monostatic ISAC setting, where the transmitted OFDM waveform probes the environment and the echo is received by the sensing receiver. Each frame consists of $N$ subcarriers and $M$ OFDM symbols, yielding $L=NM$ TF resource elements.

Given the normalized communication representation $\mathbf{z}_{{\rm c},t}$ and TF precoder $\mathbf{P}_t\in\mathbb{C}^{L\times Q}$, transmit signal processing forms the shared OFDM resource-grid symbols as
\begin{align}
    \mathbf{x}_t
    &=
    \mathbf{P}_t\mathbf{z}_{{\rm c},t}
    \in\mathbb{C}^{L},
    \label{eq:case_shared_ofdm_sequence}
\end{align}
\revised{Here, $\mathbf{x}_t$ is the vectorized TF resource grid. The matrix operation combines the continuous-valued feature-to-symbol mapping with precoding and resource mapping.} OFDM modulation, cyclic-prefix processing, and the corresponding receiver transforms are represented through the equivalent TF-domain observation models below.

For the communication branch, the physical channel in Eq.~(\ref{eq:physical_communication_channel}) is modeled as a diagonal OFDM channel
$\mathbf{H}_{{\rm c},t}=\operatorname{diag}(\mathbf{h}_t)\in\mathbb{C}^{L\times L}$, yielding
\begin{align}
    \mathbf{y}_{{\rm c},t}
    =
    \mathbf{H}_{{\rm c},t}\mathbf{x}_t
    +
    \mathbf{n}_{{\rm c},t}
    =
    \operatorname{diag}
    \left(
        \mathbf{h}_t
    \right)
    \mathbf{P}_t
    \mathbf{z}_{{\rm c},t}
    +
    \mathbf{n}_{{\rm c},t}.
    \label{eq:case_communication_observation}
\end{align}
Here, $\mathbf{h}_t\in\mathbb{C}^{L}$ is the vectorized communication-channel response, and $\widehat{\mathbf{H}}_{{\rm c},t}=\operatorname{diag}(\widehat{\mathbf{h}}_t)$ denotes the corresponding estimated channel matrix available to joint S2P design. This estimate supplies the communication-channel information in $\boldsymbol{\kappa}_{\rm phy}$.

For the sensing branch, $\mathbf{x}_t$ is reshaped as $\mathbf{X}_t\in\mathbb{C}^{N\times M}$. For a single target with reflection coefficient $\alpha_t$, delay $\tau_t^{\rm phy}$, and Doppler shift $\nu_t^{\rm phy}$, the monostatic OFDM sensing channel is modeled as \cite{braun2011singleTargetOFDMRadar}
\begin{align}
    H_{{\rm s},t}[n,m]
    =
    \alpha_t
    e^{-\mathrm{j}2\pi n\Delta f\tau_t^{\rm phy}}
    e^{\mathrm{j}2\pi mT_{\rm sym}\nu_t^{\rm phy}},
    \label{eq:case_ofdm_sensing_channel}
\end{align}
and the corresponding sensing observation in Eq.~(\ref{eq:physical_sensing_channel}) is
\begin{align}
    Y_{{\rm s},t}[n,m]
    =
    H_{{\rm s},t}[n,m]X_t[n,m]
    +
    N_{{\rm s},t}[n,m],
    \label{eq:case_ofdm_sensing_observation}
\end{align}
for $n=0,\ldots,N-1$ and $m=0,\ldots,M-1$, where $\Delta f$ is the subcarrier spacing and $T_{\rm sym}$ is the OFDM symbol duration. The model assumes a sufficient cyclic prefix and negligible intercarrier interference. The vector $\mathbf{y}_{{\rm s},t}=\operatorname{vec}(\mathbf{Y}_{{\rm s},t})$ provides the sensing observation used by \revised{sensing signal processing and the sensing semantic decoder}.

\subsection{Modular TF-Precoder Optimization}
\label{subsec:case_semantic_to_physical_instantiation}

The modular parameter partition specifies which functions in the general architecture are optimized in this realization. Let $\boldsymbol{\omega}_{\rm P}$ denote the parameters of the TF-precoding policy. The optimized and fixed parameter subsets are
\begin{align}
    \begin{aligned}
    \boldsymbol{\omega}_{\rm o}
    =
    \left\{
        \boldsymbol{\omega}_{\rm P}
    \right\},
    \quad
    \boldsymbol{\omega}_{g}
    \equiv
    \boldsymbol{\omega}_{\rm P},
    \quad
    \boldsymbol{\omega}_{\rm f}
    =
    \left\{
        \boldsymbol{\Theta}_{\rm c}^{\rm e},
        \boldsymbol{\Theta}_{\rm s}^{\rm e},
        \boldsymbol{\Theta}_{\rm c}^{\rm d},
        \boldsymbol{\Theta}_{\rm s}^{\rm d}
    \right\}.
    \end{aligned}
    \label{eq:case_modular_parameter_partition}
\end{align}
During offline training in this OFDM realization, only joint S2P design is optimized through $\boldsymbol{\omega}_{\rm P}$. The MiniLM model and classifier $F_{\rm c}$ remain frozen, while the non-neural processing blocks in Fig.~\ref{fig:ofdm_realization} use fixed rules. Each policy is trained for the specified OFDM dimensions, power budget, and non-waveform resource configuration. \revised{The selected TF precoder $\mathbf{P}_t$ realizes $\mathbf{a}_{\rm P}$, and transmit signal processing applies it to $\mathbf{z}_{{\rm c},t}$.}

This subsection uses the modular task-loss interface in \eqref{eq:semantic_to_physical_mapping} for TF-precoder optimization. The communication physical-quality variable $\mathbf{C}_{{\rm c},t}\in\mathbb{C}^{Q\times Q}$ is the latent-recovery error covariance of the linear minimum mean-square error (LMMSE) receiver. The sensing variable $\mathbf{C}_{{\rm s},t}\in\mathbb{R}^{2\times2}$ is the CRB matrix for the normalized delay--Doppler state, with the unknown complex reflection coefficient treated as a nuisance parameter. The task modules provide $\boldsymbol{\rho}_t$ and $\mathbf{W}_{{\rm s},t}$ to construct the task-aware proxies
\begin{align}
    D_{{\rm c},t}^{\rm sem}
    =
    \widehat{\Gamma}_{\rm c}
    \left(
        \mathbf{C}_{{\rm c},t};
        \boldsymbol{\rho}_t
    \right)
    =
    \sum_{q=1}^{Q}
    \rho_{q,t}\overline{c}_{q,t},
    \label{eq:case_communication_semantic_constraint}\\
    D_{{\rm s},t}^{\rm sem}
    =
    \widehat{\Gamma}_{\rm s}
    \left(
        \mathbf{C}_{{\rm s},t};
        \mathbf{W}_{{\rm s},t}
    \right)
    =
    \operatorname{tr}
    \left(
        \mathbf{W}_{{\rm s},t}
        \mathbf{C}_{{\rm s},t}
    \right).
    \label{eq:case_sensing_semantic_proxy}
\end{align}
Here, $\overline{c}_{q,t}=[\mathbf{C}_{{\rm c},t}]_{q,q}/(v_q+\varepsilon_{\rm num})$ is the variance-normalized recovery error of component $q$, where $v_q=\mathbb{E}_{\rm tr}[|z_{{\rm c},q}|^2]$ is its training-set variance and $\varepsilon_{\rm num}>0$ provides numerical stability. The communication proxy weights latent-recovery errors by classifier sensitivity, while the sensing proxy weights delay--Doppler uncertainty by the local sensitivity of a smooth arrival-risk score. These task-aware physical proxies guide precoder optimization. Their derivations are given in Appendix~\ref{app:ofdm_proxy_construction}, and their relationship to the final losses $\mathcal{L}_{\rm c}^{\rm sem}$ and $\mathcal{L}_{\rm s}^{\rm sem}$ is evaluated in Fig.~\ref{fig:task_metric_validation}.

In this realization, joint S2P design uses $\boldsymbol{\rho}_t$, derived from $\mathbf{z}_{{\rm c},t}$ through the classifier $F_{\rm c}$, as its communication-side input. The available design information is represented by
\begin{align}
    \boldsymbol{\chi}_t
    =
    \left\{
    \widehat{\mathbf{H}}_{{\rm c},t},
    \boldsymbol{\rho}_t,
    \mathbf{q}_{{\rm s},t}
    \right\}.
    \label{eq:case_policy_information}
\end{align}
The parameter $\beta\in[0,1]$ specifies the communication--sensing priority. The TF-precoding policy $\pi_{\boldsymbol{\omega}_{\rm P}}$ realizes the joint S2P function $g_{\boldsymbol{\omega}_{g}}$ for this OFDM system. It maps the pair $(\boldsymbol{\chi}_t,\beta)$ to the TF precoder $\mathbf{P}_t =\pi_{\boldsymbol{\omega}_{\rm P}} \left(\boldsymbol{\chi}_t,\beta\right)$ used in \eqref{eq:case_shared_ofdm_sequence}.

Let
$\overline{D}_{i}^{\rm sem}(\pi)
=
\mathbb{E}_{\mathcal{D}}
[D_{i,t}^{\rm sem}(\pi)]$
denote the average proxy over the offline training distribution
$\mathcal{D}$. For the given communication requirement $\epsilon_{\rm c}$ expressed in the calibrated proxy domain,
the TF-precoder design is formulated as
\begin{align}
    \min_{\pi} \quad
     \overline{D}_{\rm s}^{\rm sem}(\pi),
    \quad \mathrm{s.t.}\quad
    \overline{D}_{\rm c}^{\rm sem}(\pi) \leq \epsilon_{\rm c},
    \nonumber\\[-1mm]
    \frac{1}{L}
    \left\|
        \pi(\boldsymbol{\chi}_t,\beta)
    \right\|_{F}^{2}
    = P_{\max},
    \quad \forall t.
    \label{eq:case_sequence_precoding_problem}
\end{align}
For offline optimization, the corresponding weighted formulation under the same transmit-power constraint is
\begin{align}
    \min_{\pi}
    \quad
    \beta
    \overline{D}_{\rm c}^{\rm sem}(\pi)
    +
    (1-\beta)
    \overline{D}_{\rm s}^{\rm sem}(\pi),
    \qquad
    \beta\in[0,1].
    \label{eq:case_weighted_proxy_scalarization}
\end{align}

\subsection{Learning-Based S2P Policy}
\label{subsec:sequence_precoder_architecture}

The TF-precoder design depends on the channel state and task information through nonlinear task-loss proxies, giving a high-dimensional and generally non-convex optimization problem. Solving it at each transmission interval requires repeated LMMSE and Fisher information matrix (FIM) calculations. We therefore implement the S2P policy $\pi_{\boldsymbol{\omega}_{\rm P}}$ using a neural generator $G_{\boldsymbol{\omega}_{\rm P}}$ followed by power normalization. The generator extracts features from the estimated communication channel $\widehat{\mathbf{H}}_{{\rm c},t}$. It combines these features with $\boldsymbol{\rho}_t$, $\mathbf{q}_{{\rm s},t}$, and $\beta$ to produce the complex-valued unnormalized precoder $\widetilde{\mathbf{P}}_t$. The unnormalized and normalized precoders are given by
\begin{align}
    \widetilde{\mathbf{P}}_t
    =
    G_{\boldsymbol{\omega}_{\rm P}}
    \left(
        \boldsymbol{\chi}_t,\beta
    \right),
    \qquad
    \mathbf{P}_t
    =
    \sqrt{L P_{\max}}\,
    \widetilde{\mathbf{P}}_t
    \left\|
        \widetilde{\mathbf{P}}_t
    \right\|_{F}^{-1}.
    \label{eq:learning_sequence_precoder_mapping}
\end{align}
The normalization enforces the transmit-power constraint in \eqref{eq:case_sequence_precoding_problem}.

During offline training, $\beta_i$ is sampled from $\mathcal{U}[0,1]$ to cover different communication--sensing operating priorities. For each sample $i$, the policy produces $\mathbf{P}_i=\pi_{\boldsymbol{\omega}_{\rm P}}(\boldsymbol{\chi}_i,\beta_i)$ through \eqref{eq:learning_sequence_precoder_mapping}. Following the weighted design in \eqref{eq:case_weighted_proxy_scalarization}, the policy parameters are optimized over a mini-batch of size $B$ as
\begin{align}
    \min_{\boldsymbol{\omega}_{\rm P}}
    \quad
    \widehat{\mathcal{L}}_{\rm P}
    =
    B^{-1}
    \sum_{i=1}^{B}
    \left[
        \beta_i
        \widetilde{D}_{{\rm c},i}^{\rm sem}
        \left(
            \mathbf{P}_i
        \right)
        +
        (1-\beta_i)
        \widetilde{D}_{{\rm s},i}^{\rm sem}
        \left(
            \mathbf{P}_i
        \right)
    \right].
    \label{eq:sequence_precoder_training_loss}
\end{align}
Here, $\widetilde{D}_{b,i}^{a}$ denotes the scale-normalized per-sample proxy for branch $b\in\{{\rm c},{\rm s}\}$ and type $a\in\{{\rm sem},{\rm phy}\}$. Equation~\eqref{eq:sequence_precoder_training_loss} uses the two task-aware semantic proxies. The gradient of $\widehat{\mathcal{L}}_{\rm P}$ is backpropagated through the proxy calculations and the precoder mapping to update $\boldsymbol{\omega}_{\rm P}$.

Online, the policy uses the optimized parameters $\boldsymbol{\omega}_{\rm P}$ to map $(\boldsymbol{\chi}_t,\beta)$ to $\mathbf{P}_t$ in one forward pass, implementing joint S2P design. Transmit signal processing then applies the selected precoder through \eqref{eq:case_shared_ofdm_sequence} and the OFDM operations. \revised{The design function can also be realized by numerical optimization, unfolded optimization, reinforcement learning, or other physical-design methods.}

\section{Simulation Results}

This section evaluates the task relevance and operating behavior of the proposed SemISAC design. After describing the simulation configuration and comparison methods, we examine proxy validity, performance across transmit power and task requirements, robustness to imperfect semantic information, and task-dependent feature protection.

\subsection{Simulation Settings}

The simulations consider an OFDM system with $N=M=16$ and a communication latent dimension $Q=8$. Unless otherwise specified, the carrier frequency and subcarrier spacing are $f_c=28$ GHz and $\Delta f=120$ kHz, respectively. The communication channel contains three paths, the reference noise power is $\sigma_0^2=1$, and the normalized mean-square error (NMSE) of the communication channel state information (CSI) is $-10$ dB.

For semantic communication, we use the four-class AG News dataset \cite{zhang2015textUnderstanding}, with 20,000 training, 5,000 validation, and 7,600 testing samples. The communication component errors are normalized by feature variances estimated from the normalized AG News training latents, with $\varepsilon_{\rm num}=10^{-8}$ providing numerical stability. This normalization accounts for the different variances of the communication latent components. For semantic sensing, the binary arrival-risk task uses a predefined 10-m safety boundary and a safety horizon of $T_{\rm safe}=2$ s. The target distance and speed are uniformly distributed over $[15,60]$ m and $[5,25]$ m/s, respectively, and $\sigma_T=0.5$ s controls the transition width of the smooth arrival-risk score.

Each learned method is trained independently for every power condition for 200 epochs, with $\beta_i$ sampled from $\mathcal{U}[0,1]$ during training. Training uses AdamW \cite{loshchilov2019adamW} with a learning rate of $10^{-3}$, weight decay of $10^{-5}$, gradient norm clipping at 10, and a batch size of 128. During training, the sensing prior is a transmitter-side prediction available before the current probing interval. We model it by perturbing the delay and Doppler values used to construct $\mathbf{W}_{{\rm s},i}$ with Gaussian errors whose standard deviations equal 5\% of the corresponding parameter ranges. The transmit-power studies sweep $P_{\max}\in\{-10,-5,0,5,10,15\}$ dB, while the trade-off studies sweep $\beta\in\{0,0.1,\ldots,1\}$. For comparison, we consider four benchmarks spanning uniform precoding, conventional physical optimization, and single-task semantic conditioning.
\begin{itemize}
    \item Benchmark 1 (Uniform Precoding). The first $Q$ columns of a DFT matrix construct the precoder. This benchmark satisfies the same total-power constraint and does not adapt to the channel or task priors.
    \item Benchmark 2 (Conventional Physical ISAC). Following the deep-learning-based physical ISAC method in \cite{xiaoqi2023learnCRB}, the neural precoder is trained using the $\beta_i$-weighted physical objective $\beta_i\widetilde{D}_{{\rm c},i}^{\rm phy}+(1-\beta_i)\widetilde{D}_{{\rm s},i}^{\rm phy}$.
    \item Benchmark 3 (Semantic Communication Aware ISAC). The communication importance $\boldsymbol{\rho}_i$ is used to optimize $\widetilde{D}_{{\rm c},i}^{\rm sem}$, while the sensing design retains the physical objective $\widetilde{D}_{{\rm s},i}^{\rm phy}$.
    \item Benchmark 4 (Semantic Sensing Aware ISAC). The sensing matrix $\mathbf{W}_{{\rm s},i}$ is used to optimize $\widetilde{D}_{{\rm s},i}^{\rm sem}$, while the communication design retains equal feature importance and the physical objective $\widetilde{D}_{{\rm c},i}^{\rm phy}$.
\end{itemize}

\subsection{Validation and Performance of Semantic ISAC}

\begin{figure}[tb]
    \centering
    \includegraphics[width=\linewidth,trim=0.00bp 13.71bp 0.00bp 27.45bp,clip]{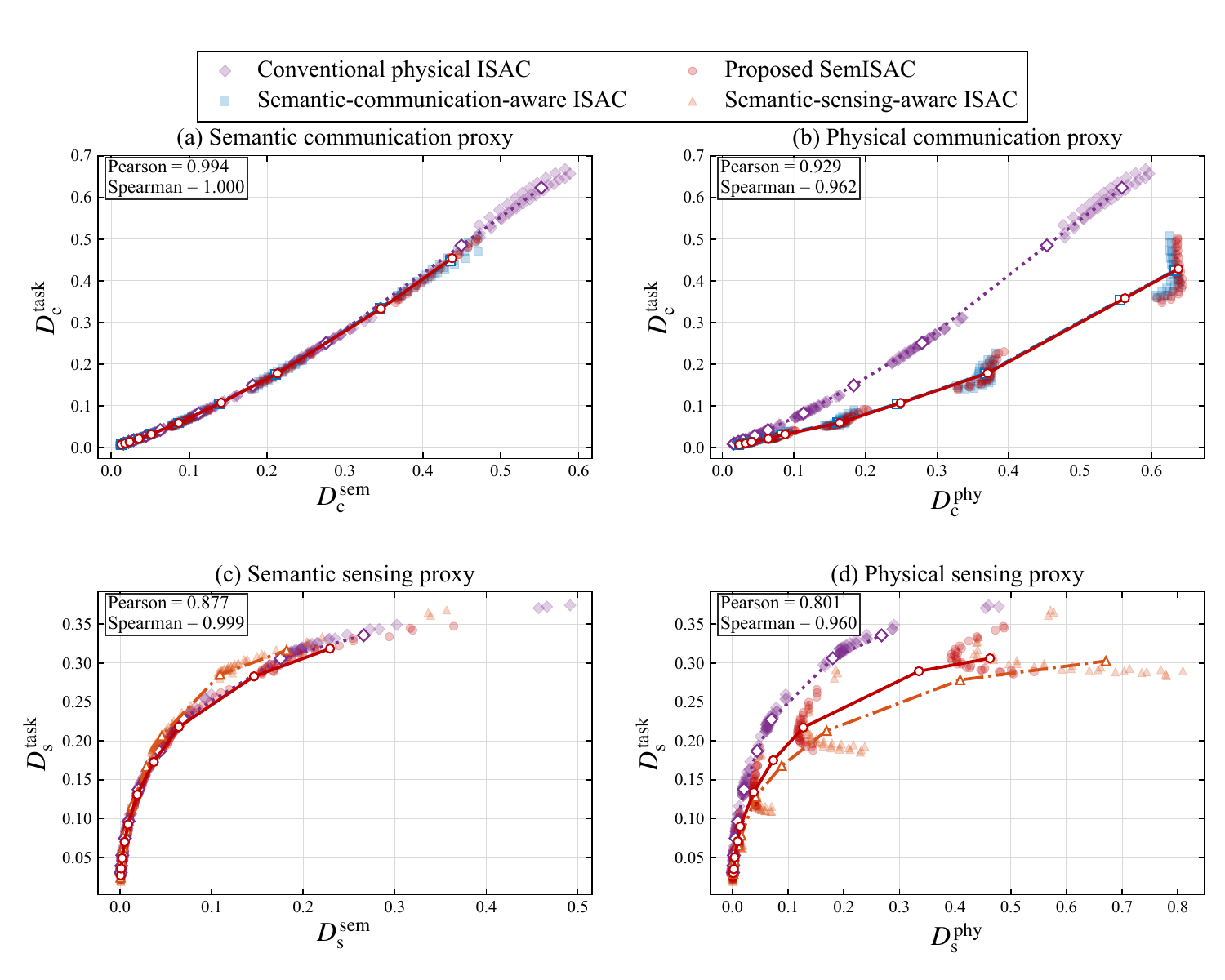}
    \caption{Validation of the semantic and physical proxies against downstream task errors.}
    \label{fig:task_metric_validation}
\end{figure}

We first compare the semantic and physical proxies with downstream communication and sensing task errors in Fig.~\ref{fig:task_metric_validation}. Here, $D_{\rm c}^{\rm task}$ denotes the classifier-output KL divergence, while $D_{\rm s}^{\rm task}$ denotes the risk-decision error probability. For communication, Fig.~\ref{fig:task_metric_validation}(a) shows an almost monotonic relationship between the normalized semantic proxy and $D_{\rm c}^{\rm task}$. The physical proxy in Fig.~\ref{fig:task_metric_validation}(b) has a weaker association with this task error and greater separation across methods, suggesting that similar unweighted physical distortions can correspond to different classifier-output KL divergences. This behavior is consistent with semantic communication designs protecting latent components according to their sample-dependent classification relevance, whereas methods without communication importance weight all components uniformly. Overall, the semantic proxy is more strongly associated with communication task error over the evaluated operating points.

For sensing, Fig.~\ref{fig:task_metric_validation}(c) shows that the semantic proxy preserves an almost monotonic ordering of the arrival-risk decision error, although the proxy and nonlinear task error remain numerically distinct. The physical proxy in Fig.~\ref{fig:task_metric_validation}(d) has a weaker association with this error and more pronounced offsets across methods. The semantic proxy weights delay--Doppler uncertainty directions by their relevance to the arrival-risk decision, whereas the physical proxy weights all sensing parameter errors uniformly.

\begin{figure}[tb]
    \centering
    \includegraphics[width=\linewidth,trim=0.00bp 13.71bp 0.00bp 32.69bp,clip]{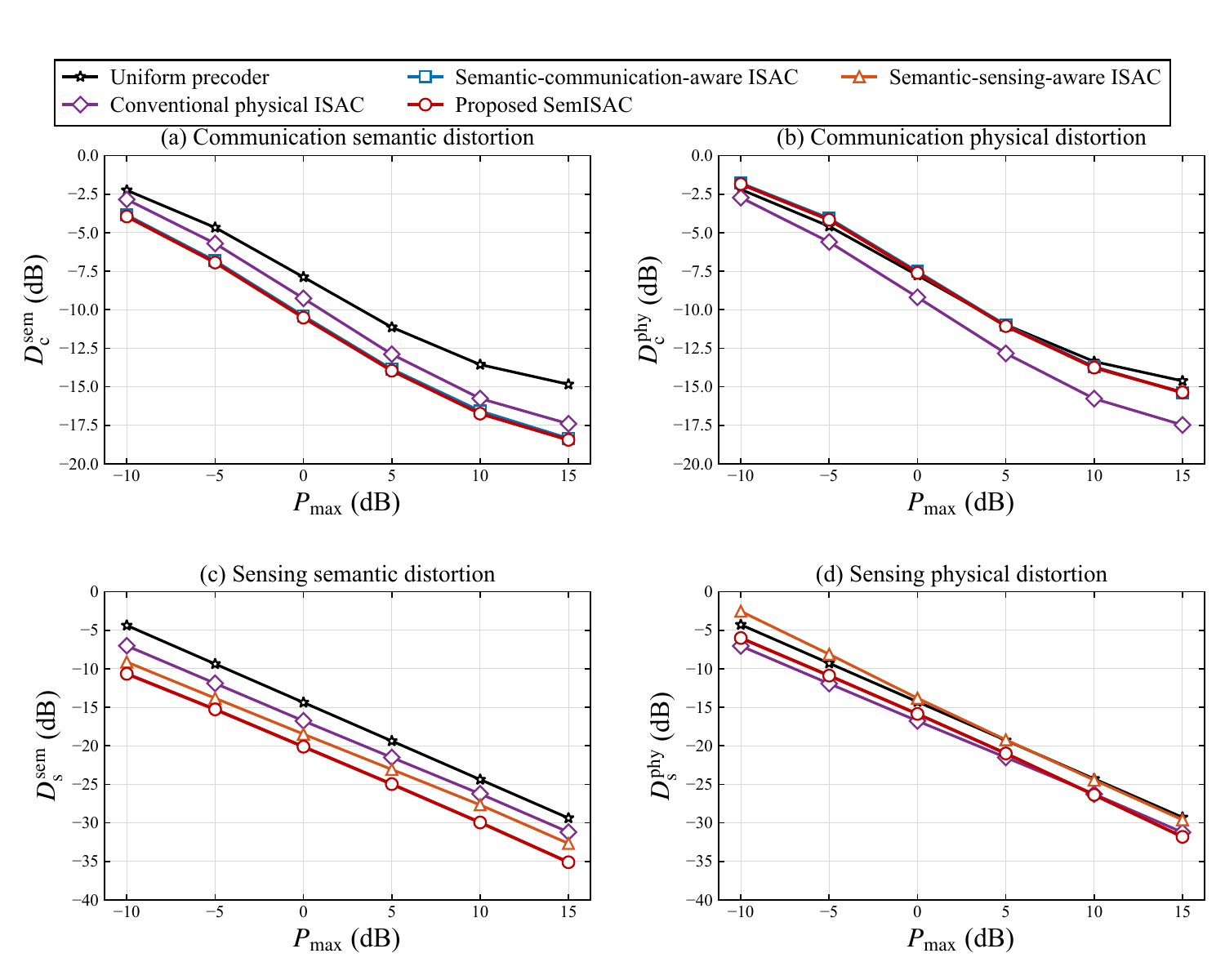}
    \caption{Communication and sensing proxy distortions in dB versus maximum transmit power.}
    \label{fig:transmit_power_metrics}
\end{figure}

We then examine how the maximum transmit power affects the four communication and sensing proxy distortions under balanced task priority and perfect semantic priors. Fig.~\ref{fig:transmit_power_metrics} reports all four distortions in dB. For communication, Fig.~\ref{fig:transmit_power_metrics}(a) shows that the semantic distortion decreases with transmit power for all methods, while proposed SemISAC and Semantic Communication Aware ISAC remain below the methods without communication importance. In contrast, Fig.~\ref{fig:transmit_power_metrics}(b) shows that Conventional Physical ISAC attains a lower physical distortion than proposed SemISAC over most of the power range. This reversal is consistent with the semantic designs protecting latent components according to their classification relevance rather than minimizing the unweighted latent error. Under the evaluated conditions, a lower communication physical distortion therefore does not necessarily correspond to a lower communication semantic distortion.

For sensing, Fig.~\ref{fig:transmit_power_metrics}(c) shows that the proposed SemISAC maintains the lowest semantic distortion across the considered power range. Fig.~\ref{fig:transmit_power_metrics}(d), however, shows that Conventional Physical ISAC has a lower physical distortion in the lower-power region, while proposed SemISAC becomes comparable and then lower toward the high-power endpoint. This reversal is consistent with the sensing task weighting emphasizing uncertainty directions relevant to the arrival-risk decision rather than all parameter errors uniformly. The sensing semantic and physical metrics therefore characterize different aspects of performance over the tested power range.

\begin{figure}[tb]
    \centering
    \includegraphics[width=0.85\linewidth,trim=0.00bp 13.71bp 0.00bp 3.56bp,clip]{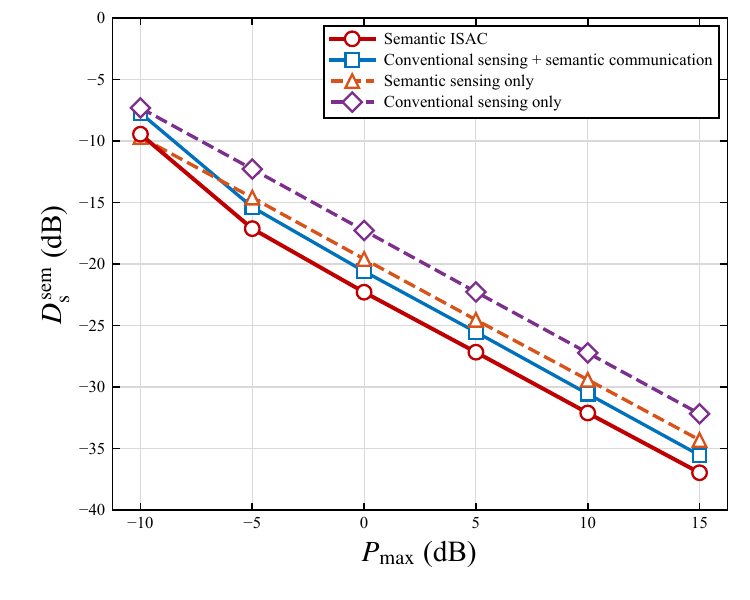}
    \caption{Sensing semantic distortion in dB versus maximum transmit power.}
    \label{fig:transmit_power_sensing}
\end{figure}

Fig.~\ref{fig:transmit_power_sensing} evaluates sensing semantic distortion in dB under the validation-derived communication proxy requirement $\epsilon_{\rm c}=0.3798$. Proposed SemISAC and Semantic Communication Aware ISAC select $\beta$ from $\{0,0.1,\ldots,1\}$ at each power to satisfy this requirement. The remaining methods are unconstrained sensing references. Distortion decreases with power for all methods. Proposed SemISAC remains below Semantic Communication Aware ISAC, with a slightly larger low-power advantage and a 1.72-dB reduction at $P_{\max}=-10$ dB. Its sensing task weighting is designed to direct waveform resources toward uncertainty directions relevant to the arrival-risk decision. Semantic Sensing Aware ISAC attains comparable distortion at the low-power endpoint but exceeds the communication requirement there. Conventional Physical ISAC has higher sensing semantic distortion throughout the power range.

\begin{figure}[tb]
    \centering
    \includegraphics[width=0.85\linewidth,trim=0.00bp 12.73bp 0.00bp 6.07bp,clip]{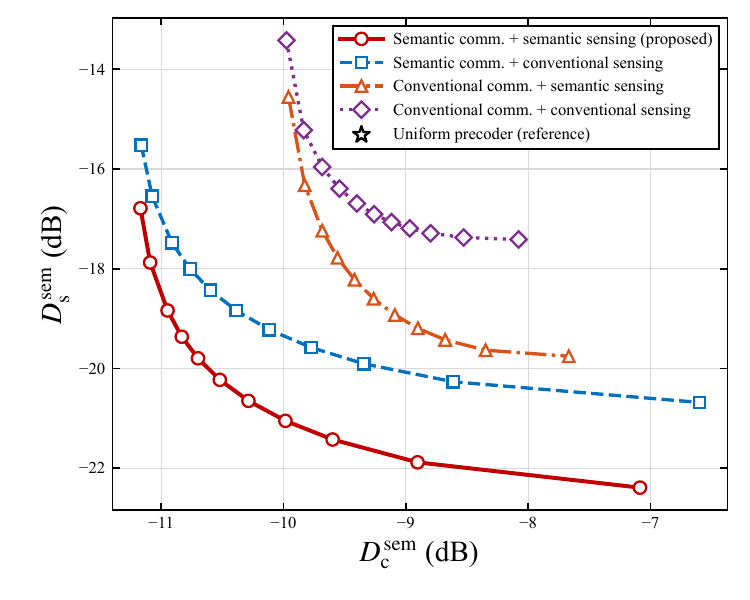}
    \caption{Communication--sensing semantic-distortion trade-off in dB at $P_{\max}=0$ dB as $\beta$ varies from 0 to 1.}
    \label{fig:semantic_distortion_tradeoff}
\end{figure}

Fig.~\ref{fig:semantic_distortion_tradeoff} evaluates the communication--sensing semantic-distortion trade-off in dB as communication priority varies. Increasing this priority generally lowers communication distortion and raises sensing distortion for each learned design, reflecting the intended shift of the shared waveform toward the communication task. Proposed SemISAC traces the lowest curve over the evaluated region. At balanced priority, it achieves lower sensing distortion and slightly lower communication distortion than Semantic Communication Aware ISAC. The single-branch semantic baselines favor the task represented by their semantic prior, whereas Conventional Physical ISAC and Uniform Precoding remain farther from the lower-left region without joint semantic conditioning. Varying $\beta$ therefore selects the operating point, while joint semantic conditioning improves the communication--sensing balance.

\begin{figure}[tb]
    \centering
    \includegraphics[width=0.45\textwidth,trim=0.00bp 12.20bp 0.00bp 8.94bp,clip]{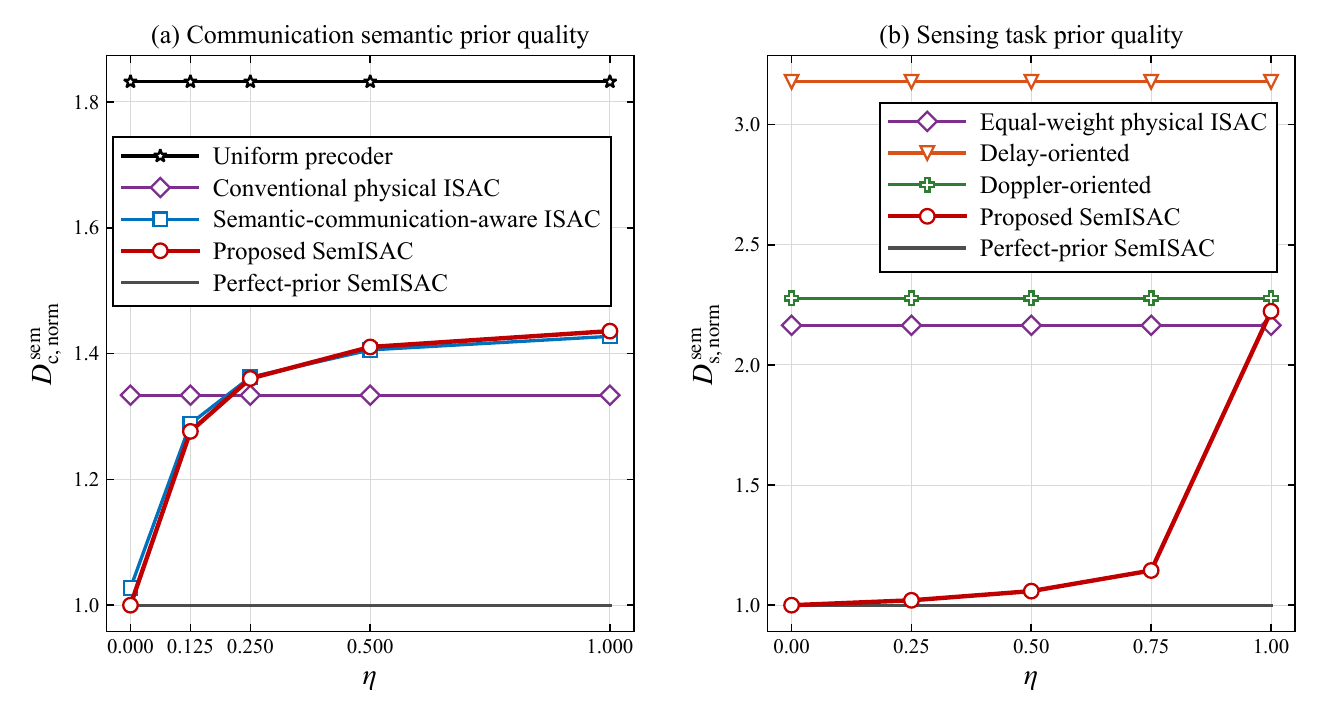}
    \caption{Robustness to semantic-prior mismatch under a balanced task priority.}
    \label{fig:semantic_prior_mismatch}
\end{figure}

For branch $i\in\{{\rm c},{\rm s}\}$, the normalized semantic distortion in Fig.~\ref{fig:semantic_prior_mismatch} is defined as $D_{i,{\rm norm}}^{\rm sem}(\eta)\triangleq D_i^{\rm sem}(\eta)/D_{i,{\rm pp}}^{\rm sem}$, where $D_{i,{\rm pp}}^{\rm sem}$ denotes the corresponding distortion of proposed SemISAC under perfect prior information. In Fig.~\ref{fig:semantic_prior_mismatch}(a), $\eta$ is the communication-prior perturbation level divided by its maximum tested value. In Fig.~\ref{fig:semantic_prior_mismatch}(b), it is the arrival-risk query replacement ratio. We next assess how communication-prior mismatch and sensing-query mismatch affect proposed SemISAC. For communication, Fig.~\ref{fig:semantic_prior_mismatch}(a) shows that the normalized semantic distortion of proposed SemISAC and Semantic Communication Aware ISAC increases as the supplied feature importance departs from the true prior. Both methods remain below Conventional Physical ISAC under small mismatch, but this ordering reverses as mismatch grows. This trend suggests that inaccurate importance values impair the protection of classification-relevant latent components.

The Perfect Prior SemISAC reference uses the unperturbed sensing requirement, and Equal Weight Physical ISAC is the conventional physical benchmark. The Doppler Oriented reference uses the Semantic Sensing Aware ISAC policy with the sensing weight $\operatorname{diag}(0.1,0.9)$. For sensing, Fig.~\ref{fig:semantic_prior_mismatch}(b) shows that proposed SemISAC remains close to Perfect Prior SemISAC under moderate query mismatch. When the arrival-risk query is fully replaced by a delay-estimation query, the normalized distortion increases to 2.22, near the Equal Weight Physical ISAC and Doppler Oriented references. The supplied sensing task weighting then emphasizes delay estimation rather than the uncertainty directions required by the arrival-risk decision.

\begin{figure}[tb]
    \centering
    \includegraphics[width=0.85\linewidth,trim=0.00bp 17.53bp 0.00bp 6.07bp,clip]{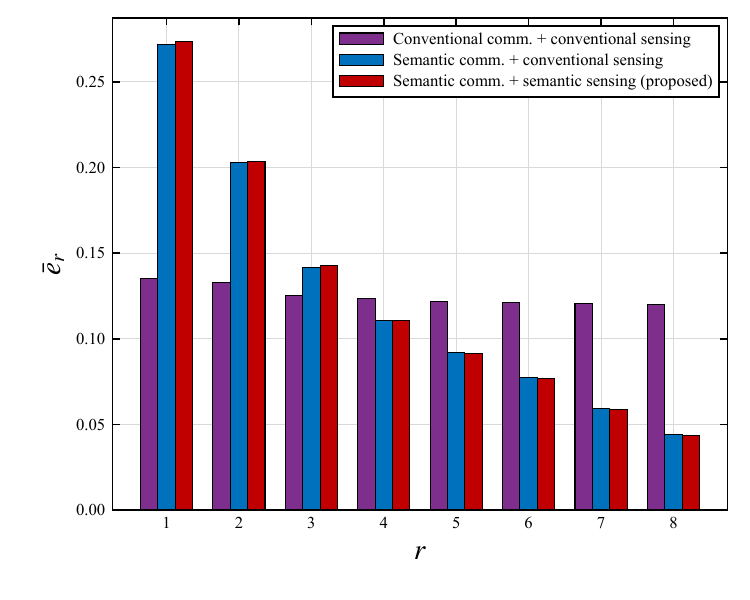}
    \caption{Mean normalized precoder-column energy $\bar e_r$ versus per-sample semantic-importance rank $r$ under a balanced task priority.}
    \label{fig:semantic_feature_protection}
\end{figure}

Finally, Fig.~\ref{fig:semantic_feature_protection} shows the energy allocation across latent features ranked by communication semantic importance under balanced task priority. The rank $r\in\{1,\ldots,Q\}$ orders features from highest to lowest classification relevance. Column energies are normalized to sum to one per sample and then averaged at each rank to obtain $\bar e_r$. Conventional Physical ISAC allocates energy almost uniformly across ranks, whereas Semantic Communication Aware ISAC and proposed SemISAC allocate progressively less energy to less relevant features. The weighted communication loss assigns larger penalties to classification-relevant latent components. The close agreement between the two semantic communication designs suggests that incorporating the sensing prior preserves the communication feature-protection pattern at the balanced operating point.

\section{Conclusion}

In this paper, we developed a general SemISAC framework that advances ISAC from physical-resource sharing toward semantic-level integration for task-relevant information transmission and acquisition. We provided an information-theoretic interpretation of the physical and task-relevant information supported by the shared waveform and formulated a general task-oriented optimization problem with E2E and modular SemISAC optimization as two realization methods. As a concrete realization, we developed a learning-based OFDM TF precoding design using task-aware semantic proxies for representative SemCom and SemS tasks. Simulations showed that these proxies track downstream task errors more closely than physical proxies. The design reduces sensing semantic distortion under a given communication requirement and improves the communication--sensing trade-off.


\appendices
\section{\revised{Derivation of the Task-Loss Interfaces for the OFDM Realization}}
\label{app:ofdm_proxy_construction}

This appendix derives the task-loss interfaces used in \eqref{eq:case_communication_semantic_constraint} and \eqref{eq:case_sensing_semantic_proxy}.

\subsection{\revised{Communication Task-Loss Interface from Classifier Sensitivity}}

The communication interface characterizes the sensitivity of the classification output to the recovery error of each transmitted latent component. For $\mathbf{z}_{{\rm c},t}\in\mathbb{C}^{Q}$, define its real representation as
\begin{align}
    \mathbf{r}_t
    =
    \left[
        \operatorname{Re}\{z_{{\rm c},1,t}\},
        \operatorname{Im}\{z_{{\rm c},1,t}\},
        \ldots,
        \operatorname{Re}\{z_{{\rm c},Q,t}\},
        \operatorname{Im}\{z_{{\rm c},Q,t}\}
    \right]^{T}
    \label{eq:case_classifier_real_representation}
\end{align}
The classifier produces logits $\boldsymbol{\ell}_t=f_{\rm c}(\mathbf{r}_t)$ and output $\mathbf{p}_t=\operatorname{softmax}(\boldsymbol{\ell}_t)$.

Let $\mathbf{J}_t$ denote the Jacobian of the classifier logits, and define $\mathbf{F}_t$ through
\begin{align}
    \mathbf{J}_t
    =
    \frac{\partial\boldsymbol{\ell}_t}
    {\partial\mathbf{r}_t^{T}},
    \qquad
    \mathbf{F}_t
    =
    \operatorname{diag}\left(\mathbf{p}_t\right)
    -
    \mathbf{p}_t\mathbf{p}_t^{T}.
    \label{eq:case_classifier_fisher}
\end{align}
The corresponding local sensitivity matrix is $\mathbf{M}_t=\mathbf{J}_t^{T}\mathbf{F}_t\mathbf{J}_t$.
Let $\widehat{\mathbf{r}}_t$ denote the real representation of the recovered latent $\widehat{\mathbf{z}}_{{\rm c},t}$ using the ordering in \eqref{eq:case_classifier_real_representation}. For the perturbation $\mathbf{e}_t=\widehat{\mathbf{r}}_t-\mathbf{r}_t$, the classifier-output KL divergence admits the local second-order approximation \cite{kullback1951informationSufficiency}
\begin{align}
    D_{\rm KL}
    \left(
        \mathbf{p}_t
        \,\Vert\,
        \widehat{\mathbf{p}}_t
    \right)
    \approx
    \frac{1}{2}
    \mathbf{e}_t^{T}
    \mathbf{M}_t
    \mathbf{e}_t.
    \label{eq:case_communication_local_kl}
\end{align}

Using the diagonal curvature terms, the relevance of the $q$th complex latent component is obtained by combining its real and imaginary parts as $\widetilde{\rho}_{q,t}=\left[\mathbf{M}_t\right]_{2q-1,2q-1}+\left[\mathbf{M}_t\right]_{2q,2q}$. The normalized relevance is
\begin{align}
    \rho_{q,t}
    =
    \frac{
        \widetilde{\rho}_{q,t}
        +
        \varepsilon_{\rho}
    }{
        \sum_{k=1}^{Q}
        \left(
            \widetilde{\rho}_{k,t}
            +
            \varepsilon_{\rho}
        \right)
    },
    \qquad
    \boldsymbol{\rho}_t
    =
    \left[
        \rho_{1,t},
        \ldots,
        \rho_{Q,t}
    \right]^{T}.
    \label{eq:case_semantic_feature_relevance}
\end{align}
Thus, $\boldsymbol{\rho}_t$ provides the componentwise task relevance used by the communication \revised{task-loss interface}.

For the communication model in \eqref{eq:case_communication_observation}, define the true and estimated effective channels as
\begin{align}
    \mathbf{A}_t
    =
    \operatorname{diag}
    \left(
        \mathbf{h}_t
    \right)
    \mathbf{P}_t,
    \qquad
    \widehat{\mathbf{A}}_t
    =
    \operatorname{diag}
    \left(
        \widehat{\mathbf{h}}_t
    \right)
    \mathbf{P}_t.
    \label{eq:case_semantic_symbol_mapping}
\end{align}
The LMMSE receiver based on the estimated channel is
\begin{align}
    \mathbf{F}_{{\rm c},t}^{\rm LMMSE}
    =
    \left(
        \widehat{\mathbf{A}}_t^{H}
        \widehat{\mathbf{A}}_t
        +
        \sigma_0^2\mathbf{I}_{Q}
    \right)^{-1}
    \widehat{\mathbf{A}}_t^{H}
    \label{eq:case_communication_lmmse_filter}
\end{align}
with $\widehat{\mathbf{z}}_{{\rm c},t}=\mathbf{F}_{{\rm c},t}^{\rm LMMSE} \mathbf{y}_{{\rm c},t}$. The resulting latent-recovery error covariance, evaluated under the true channel, is
\begin{align}
    \mathbf{C}_{{\rm c},t}
    =
    \mathbb{E}
    \left[
        \left(
            \widehat{\mathbf{z}}_{{\rm c},t}
            -
            \mathbf{z}_{{\rm c},t}
        \right)
        \left(
            \widehat{\mathbf{z}}_{{\rm c},t}
            -
            \mathbf{z}_{{\rm c},t}
        \right)^{H}
    \right]
    \nonumber\\
    =
    \left(
        \mathbf{F}_{{\rm c},t}^{\rm LMMSE}
        \mathbf{A}_t
        -
        \mathbf{I}_{Q}
    \right)
    \left(
        \mathbf{F}_{{\rm c},t}^{\rm LMMSE}
        \mathbf{A}_t
        -
        \mathbf{I}_{Q}
    \right)^{H}
    \nonumber\\
    \quad
    +
    \sigma_0^2
    \mathbf{F}_{{\rm c},t}^{\rm LMMSE}
    \left(
        \mathbf{F}_{{\rm c},t}^{\rm LMMSE}
    \right)^{H}.
    \label{eq:case_communication_error_covariance}
\end{align}

Combining $\rho_{q,t}$ with the variance-normalized component errors $\overline{c}_{q,t}$ defined after \eqref{eq:case_communication_semantic_constraint} gives the task-aware communication proxy. The corresponding task-independent physical proxy is $D_{{\rm c},t}^{\rm phy}=Q^{-1}\operatorname{tr}(\mathbf{C}_{{\rm c},t})$.

\subsection{\revised{Sensing Task-Loss Interface from Arrival-Risk Sensitivity}}

The normalized delay--Doppler state and arrival-time mapping are defined in Section~\ref{sec:ofdm_case_study} and \eqref{eq:case_sensing_arrival_time_mapping}.

A differentiable risk score is used to characterize the local task sensitivity as
\begin{align}
    u_{\rm s}
    \left(
        \boldsymbol{\theta}
    \right)
    =
    \operatorname{sigmoid}
    \left(
        \frac{
            T_{\rm safe}
            -
            t_{\rm arr}
            \left(
                \boldsymbol{\theta}
            \right)
        }{
            \sigma_T
        }
    \right).
    \label{eq:case_sensing_task_information}
\end{align}

Since the instantaneous state is unavailable before probing, \revised{the requirement function $E_{\rm s}$ evaluates} the local task sensitivity at the transmitter-available prior $\overline{\boldsymbol{\theta}}_t$. The task gradient and normalized sensing task-sensitivity matrix are
\begin{align}
    \mathbf{G}_{{\rm s},t}
    =
    \left.
    \frac{
        \partial
        u_{\rm s}
        \left(
            \boldsymbol{\theta}
        \right)
    }{
        \partial
        \boldsymbol{\theta}^{T}
    }
    \right|_{
        \boldsymbol{\theta}
        =
        \overline{\boldsymbol{\theta}}_t
    },
    \label{eq:case_sensing_task_gradient}\\
    \mathbf{W}_{{\rm s},t}
    =
    \frac{
        \mathbf{G}_{{\rm s},t}^{T}
        \mathbf{G}_{{\rm s},t}
    }{
        \mathbb{E}_{\rm tr}
        \left[
            \operatorname{tr}
            \left(
                \mathbf{G}_{\rm s}^{T}
                \mathbf{G}_{\rm s}
            \right)
        \right]
        +
        \varepsilon_{\rm s}
    }.
    \label{eq:case_sensing_task_weighting_matrix}
\end{align}
Here, $\varepsilon_{\rm s}$ is a numerical-stability constant. The resulting matrix defines the sensing requirement used in Section~\ref{sec:ofdm_case_study}.

For final task evaluation, \revised{sensing signal processing estimates} delay and Doppler from \eqref{eq:case_ofdm_sensing_observation} using the following matching metric
\begin{align}
    \Lambda_t(\tau,\nu)
    =
    \Bigg|
    \sum_{n=0}^{N-1}
    \sum_{m=0}^{M-1}
    &Y_{{\rm s},t}[n,m]
    X_t^{*}[n,m]
    \nonumber\\
    &\times
    e^{\mathrm{j}2\pi n\Delta f\tau}
    e^{-\mathrm{j}2\pi mT_{\rm sym}\nu}
    \Bigg|.
    \label{eq:case_sensing_matching_metric}
\end{align}
The delay--Doppler estimate is then obtained over the bounded search grid $\mathcal{G}_{\tau\nu}$ as
\begin{align}
    \left(
        \widehat{\tau}_t^{\rm phy},
        \widehat{\nu}_t^{\rm phy}
    \right)
    =
    \underset{
        (\tau,\nu)\in\mathcal{G}_{\tau\nu}
    }{\arg\max}
    \;
    \Lambda_t(\tau,\nu).
    \label{eq:case_sensing_grid_estimator}
\end{align}
The estimates are converted to $\widehat{\boldsymbol{\theta}}_t$ and subsequently mapped to $\widehat{U}_{{\rm s},t}$ through the arrival-time decision in \eqref{eq:case_sensing_task_output}.

For waveform optimization, $\boldsymbol{\xi}_t=[\boldsymbol{\theta}_t^{T},\alpha_{\operatorname{Re},t},\alpha_{\operatorname{Im},t}]^{T}$ represents the normalized delay--Doppler state and the unknown complex reflection coefficient. Let $\boldsymbol{\mu}_{{\rm s},t}(\boldsymbol{\xi}_t)$ denote the vectorized noise-free observation implied by \eqref{eq:case_ofdm_sensing_channel} and \eqref{eq:case_ofdm_sensing_observation}. For white complex Gaussian noise, the FIM entries are \cite{stoica1990musicMLCRB}
\begin{align}
    \left[
        \mathbf{J}_{{\rm s},t}
    \right]_{a,b}
    =
    \frac{2}{\sigma_0^2}
    \operatorname{Re}
    \left\{
        \left(
            \frac{
                \partial
                \boldsymbol{\mu}_{{\rm s},t}
            }{
                \partial\xi_{a,t}
            }
        \right)^{H}
        \frac{
            \partial
            \boldsymbol{\mu}_{{\rm s},t}
        }{
            \partial\xi_{b,t}
        }
    \right\}.
    \label{eq:case_sensing_fisher_information}
\end{align}

Partitioning the FIM with respect to the delay--Doppler state and the nuisance reflection coefficient gives
\begin{align}
    \mathbf{J}_{{\rm s},t}
    =
    \begin{bmatrix}
        \mathbf{J}_{\theta\theta,t}
        &
        \mathbf{J}_{\theta\alpha,t}
        \\
        \mathbf{J}_{\alpha\theta,t}
        &
        \mathbf{J}_{\alpha\alpha,t}
    \end{bmatrix}.
    \label{eq:case_sensing_fim_partition}
\end{align}
The equivalent FIM for the normalized delay--Doppler state and the corresponding CRB are
\begin{align}
    \mathbf{J}_{\theta\mid\alpha,t}
    =
    \mathbf{J}_{\theta\theta,t}
    -
    \mathbf{J}_{\theta\alpha,t}
    \mathbf{J}_{\alpha\alpha,t}^{-1}
    \mathbf{J}_{\alpha\theta,t},
    \label{eq:case_sensing_equivalent_fim}\\
    \mathbf{C}_{{\rm s},t}
    =
    \mathbf{J}_{\theta\mid\alpha,t}^{-1}
    =
    \left[
        \mathbf{J}_{{\rm s},t}^{-1}
    \right]_{1:2,1:2}.
    \label{eq:case_sensing_crb}
\end{align}
The task-independent sensing proxy is $D_{{\rm s},t}^{\rm phy}=\frac{1}{2}\operatorname{tr}(\mathbf{C}_{{\rm s},t})$.

The task-aware sensing interface follows from a local first-order approximation of the smooth arrival-risk score. With $\mathbf{e}_{\theta,t}=\widehat{\boldsymbol{\theta}}_t-\boldsymbol{\theta}_t$, the risk-score perturbation and its local mean-square approximation are
\begin{align}
    \Delta u_{{\rm s},t} =
    u_{\rm s}
    \left(
        \widehat{\boldsymbol{\theta}}_t
    \right)
    -
    u_{\rm s}
    \left(
        \boldsymbol{\theta}_t
    \right)
    \approx
    \mathbf{G}_{{\rm s},t}
    \mathbf{e}_{\theta,t},
    \label{eq:case_sensing_delta_method}\\
    \mathbb{E}
    \left[
        \left|
            \Delta u_{{\rm s},t}
        \right|^{2}
    \right]
    \approx
    \operatorname{tr}
    \left(
        \mathbf{G}_{{\rm s},t}^{T}
        \mathbf{G}_{{\rm s},t}
        \operatorname{Cov}
        \left(
            \mathbf{e}_{\theta,t}
        \right)
    \right).
    \label{eq:case_sensing_delta_mse}
\end{align}
Using $\mathbf{C}_{{\rm s},t}$ as the waveform-dependent covariance proxy with the normalization in \eqref{eq:case_sensing_task_weighting_matrix} yields $D_{{\rm s},t}^{\rm sem}=\operatorname{tr}(\mathbf{W}_{{\rm s},t}\mathbf{C}_{{\rm s},t})$, which recovers \eqref{eq:case_sensing_semantic_proxy} and characterizes the waveform-dependent uncertainty of the local arrival-risk score.

\bibliographystyle{IEEEtran}
\bibliography{ref}

\end{document}